\documentclass{lmcs} %%% last changed 2014-08-20
\pdfoutput=1

\usepackage[utf8]{inputenc}
\usepackage{lastpage}
\lmcsdoi{19}{4}{32}
\lmcsheading{}{\pageref{LastPage}}{}{}%
{Oct.~04,~2022}{Dec.~18,~2023}{}


\usepackage{hyperref}
\theoremstyle{plain} %\crefname{satz}{Satz}{S\"atze}
\usepackage{proof,amssymb}

\newcommand{\BS}{\mathop{\backslash}}
\newcommand{\SL}{\mathop{/}}

\newcommand{\BSU}{\BS\nolimits_U}
\newcommand{\SLU}{\SL\nolimits_{\!U}}

\newcommand{\LL}{\mathbf{L}}
\newcommand{\LS}{\mathbf{L}^{\!\mathbf{*}}}
\newcommand{\LE}{\mathbf{L}^{\!\Lambda}}

\newcommand{\Zero}{\mathbf{0}}
\newcommand{\One}{\mathbf{1}}

\newcommand{\LEwuz}{\LE{\wedge}\Zero\One}
\newcommand{\LEwuzbc}{\LEwuz bc}
\newcommand{\LEwuzK}{\LEwuz\mathrm{ItD}}
\newcommand{\LEw}{\LE{\wedge}}
\newcommand{\LLw}{\LL{\wedge}}
\newcommand{\LEDwuz}{\LE_{\BS,\SL}{\wedge}\Zero\One}
\newcommand{\LEDw}{\LE_{\BS,\SL}{\wedge}}
\newcommand{\eLEwuz}{\boldsymbol{!}\LEwuz}

\newcommand{\Fm}{\mathrm{Fm}}
\newcommand{\yields}{\to}

\newcommand{\Hc}{\mathcal{H}}
\newcommand{\Lc}{\mathcal{L}}
\newcommand{\Md}{\mathcal{M}}
\newcommand{\Pc}{\mathcal{P}}
\newcommand{\Kc}{\mathcal{K}}

\newcommand{\NN}{\mathbb{N}}

\newcommand{\Af}{\mathfrak{A}}
\newcommand{\UA}{\One_{\Af}}
\newcommand{\ZA}{\Zero_{\Af}}
\begin{document}

\title[R-Models for the Lambek Calculus with Intersection and Constants]{Relational Models for the Lambek Calculus \\ with Intersection and Constants{\rsuper*}} % \texorpdfstring{\MakeLowercase{\texttt{lmcs.cls}}}{lmcs.cls}\rsuper*\\Version of
%  2022-04-01}
\titlecomment{{\lsuper*}This article is an extended version of the conference paper presented at RAMiCS 2021.}
%OPTIONAL comment concerning the title, \eg,
%  if a variant or an extended abstract of the paper has appeared elsewhere.}
\thanks{The author is grateful to Daniel Rogozin for fruitful discussions. The work was supported by the Theoretical Physics and Mathematics
Advancement Foundation ``BASIS''}	%optional

% affiliations are numbered automatically with a, b, c (see below)
% use the optional argument to indicate the affiliation(s) of each author
% omit the argument if there is only one author, or only one affiliation
\author[S.\,L.~Kuznetsov]{Stepan L. Kuznetsov\lmcsorcid{0000-0003-0025-0133}}
%\author[B.~Name2]{Bob Name2}[a,b]
%\author[J.~Name3]{Josiah S.~Carberry\lmcsorcid{0000-0002-1825-0097}}[a]

% affiliation 1 (automatically numbered a)
\address{Steklov Mathematical Institute of RAS, 8 Gubkina St., Moscow, Russia}	%optional
% write emails for all authors having that affiliation
\email{sk@mi-ras.ru}  %optional

% affiliation 2 (automatically numbered b)
%\address{University 2, address2}	%optional
%\email{name2@email2}  %optional

%% etc.

%% required for running head on odd and even pages, use suitable
%% abbreviations in case of long titles and many authors:
%%%%%%%%%%%%%%%%%%%%%%%%%%%%%%%%%%%%%%%%%%%%%%%%%%%%%%%%%%%%%%%%%%%%%%%%%%%

%% the abstract has to PRECEDE the command \maketitle:
%% be sure not to issue the \maketitle command twice!

\begin{abstract}
  \noindent 
  We consider relational semantics (R-models) for the Lambek calculus extended with intersection and explicit constants for zero and unit. For its variant without constants and a restriction which disallows empty antecedents, Andr\'eka and Mikul\'as (1994) prove strong completeness. We show that it fails without this restriction, but, on the other hand, prove weak completeness for non-standard interpretation of constants. For the standard interpretation, even weak completeness fails. The weak completeness result extends to an infinitary setting, for  so-called iterative divisions (Kleene star under division). We also prove  strong completeness results for product-free fragments. 
\end{abstract}

\maketitle

\section*{Introduction}

The {\em Lambek calculus} $\LL$ was originally introduced back in 1958 for modelling natural language syntax~\cite{Lambek1958}. From a modern point of view, the Lambek calculus is considered as the most basic associative substructural logic, being the algebraic logic of {\em residuated semigroups}~\cite{GalatosRLbook}. On the other hand, as noticed by Abrusci~\cite{Abrusci1990}, the Lambek calculus can be viewed as a non-commutative, multiplicative-only variant of Girard's~\cite{Girard1987} {\em linear logic.}

A residuated semigroup is a semigroup with a partial order $\preceq$ and division operations $\BS$ and $\SL$ obeying the following equivalences:
\[
b \preceq a \BS c \iff a \cdot b \preceq c \iff a \preceq c \SL b.
\]
Notice that here divisions and multiplication are connected via partial order, not equality. The idea of such divisions, also called {\em residuals,} goes back to the works of Krull~\cite{Krull1924} and Ward and Dilworth~\cite{WardDilworth1939}.

Formulae of the Lambek calculus are built from a variables using the operations of residuated semigroups: $\cdot$, $\BS$, $\SL$. From the point of view of first-order logic, such formulae are terms in the signature of residuated semigroups. Furthermore, $\LL$ operates expressions of the form $A \to B$, where $A$ and $B$ are formulae. Here $\to$ corresponds to $\preceq$, thus, from the first-order point of view, these are atomic formulae. We call such expressions {\em atomic sequents.}

The set of theorems of~$\LL$ can be semantically described as the set of all sequents which are true on any residuated semigroup under any interpretation of variables. Syntactically, $\LL$ is axiomatised as a Gentzen-style sequent calculus~\cite{Lambek1958}. This calculus is presented in Section~\ref{S:calculus} below. Derivable objects of this sequent calculus are {\em sequents} of the form $A_1, \ldots, A_n \to B$, which generalise atomic sequents defined above.

Natural language applications, connections to linear logic, and algebraic interpretations on residuated structures suggest extending the Lambek calculus with extra operations.
Adding new operations and constants to the algebraic construction of residated semigroups results in extending the Lambek calculus. In this article, we focus on extending~$\LL$ with the following:
\begin{itemize}
 \item the unit constant~$\One$, which is the unit for multiplication;
 \item the zero constant~$\Zero$, which is the smallest element for $\preceq$ (using division operations one can also prove that $\Zero$ is the zero element for multiplication);
 \item the intersection, or {\em meet} operation $\wedge$: $a \wedge b = \inf_{\preceq} \{ a, b \}$; meet turns the preorder into a lower semilattice.
\end{itemize}
Note that extending $\LL$ with $\One$ and $\wedge$ is also due to Lambek~\cite{Lambek1961,Lambek1969}.

An intricate issue arises when extending~$\LL$ with the unit constant~$\One$: unlike most others, this extension is not a {\em conservative} one. Consider the atomic sequent $(p \BS p) \BS q \to q$, which belongs to the original language of $\LL$. This atomic sequent is true on all residuated monoids (i.e., residuated semigroups with the unit). However, it fails to be true on all residuated semigroups (which form a larger class of algebraic structures). 

Thus, even in the language without~$\One$, the algebraic logic of all residuated monoids differs from~$\LL$. This logic is called {\em the Lambek calculus allowing empty antecedents}~\cite{Lambek1961} and is commonly denoted by $\LS$. We prefer, however, an alternative notation, $\LE$, in order to avoid notation conflicts with Kleene star which we shall use, in the form of iterative divisions, in one of the sections of this article. The term ``... allowing empty antecedents'' comes from the Gentzen-style formulation, see Section~\ref{S:calculus} below.

Residuated algebraic structures provide a very abstract and general framework for semantics of the Lambek calculus, its variants and extensions. This is traditional {\em provability} semantics,   which interprets theoremhood and entailment, not proofs. More modern {\em semantics of proofs} (see, e.g.,~\cite{Almeida1996,Blute1996,Coecke2013,dePaiva2018}) are beyond the scope of this article. 

As usual, proving completeness results for abstract algebraic semantics is a simple application of the Lindenbaum -- Tarski construction (cf.~\cite{BlokPigozzi}). This construction gives {\em strong completeness,} that is, completeness not only for theoremhood, but also for the entailment relation (derivability vs.\ semantic entailment from sets of hypotheses).
More interesting issues arise when one starts {\em concretising} algebraic semantics, that is, considers specific classes of algebras. Completeness theorems for such classes, if they hold, are non-trivial results. We shall also see cases where completeness fails, or holds only in the weak sense (for theoremhood, not for entailment).

This article concentrates on {\em relational models,} or {\em R-models.} In R-models, variables and formulae are interpreted as binary relations over a non-empty set $W$.
The set $\Pc(W \times W)$ of {\em all} binary relations over $W$ has a natural structure of residuated monoid. Multiplication is relation composition:
\[
 R \cdot S = R \circ S = \{ (x,z) \mid (\exists y \in W) \, ((x,y) \in R \mbox{ and } (y,z) \in S) \}.
\]
Divisions are defined as follows:
\begin{align*}
& R \BS S = \{ (y,z) \mid (\forall x \in W) \, ((x,y) \in R \Rightarrow (x,z) \in S \},\\
& S \SL R = \{ (x,y) \mid (\forall z \in W) \,
((y,z) \in R \Rightarrow (x,z) \in S \}.
\end{align*}
The r\^{o}le of the unit is played by the diagonal relation $\delta = \{ (x,x) \mid x \in W \}$. 
The preorder is the set inclusion relation $\subseteq$. Meet is set-theoretic intersection, and zero is the empty relation $\varnothing$.

Interpretations of the Lambek calculus on residuated monoids of all binary relations on a non-empty set $W$ are called {\em unrelativised} or {\em square} R-models. Unrelativised R-models give natural semantics for $\LE$; the completeness theorem was proved by Andr\'{e}ka and Mikul\'{a}s~\cite{AndrekaMikulas1994}. Notice that the argument used for proving this completeness result does not easily extend to $\Zero$, $\One$, and $\wedge$. Issues with R-models for these extensions  form the main topic of this article.

More precisely, constants~$\Zero$ and~$\One$ ruin completeness w.r.t.\ square R-models, even in the weak sense (see Section~\ref{S:constants} below). In order to overcome this, we introduce variations of R-models with non-standard interpretation of constants.

As for the extension of $\LE$ with intersection, Mikul\'{a}s~\cite{Mikulas2015Synthese,Mikulas2015Studia} proves completeness w.r.t.\ square R-models, but only in the weak sense. We show (Section~\ref{S:counterexample}) that this is essential, and strong completeness fails. On the other hand, we strengthen Mikul\'{a}s' result by adding constants~$\Zero$ and~$\One$, with non-standard interpretations. Our argument also gives an alternative proof of Mikul\'{a}s' result (without constants), which happens to be extendable to an infinitary extension of the calculus, with so-called iterative divisions (Section~\ref{S:iterative}).
Finally, we show that strong completeness restores for the variant of $\LE$ with meet (intersection) but without product (Section~\ref{S:strong}). Thus, the results presented in this article fill some gaps in the theory of R-models for extensions of~$\LE$.

In order to provide adequate semantics for the original Lambek calculus~$\LL$, one needs to relax the definition of R-model. This is done by {\em relativising} it. Namely, instead of all binary relations on $W$ one may now consider only subsets of a fixed transitive relation $U$, which is called the ``universal'' one. By transitivity,  
composition (product) keeps this relativisation in place; for divisions, again, one needs to impose it implicitly:
\[
R \BSU S = \{ (y,z) \in U \mid (\forall x \in W) \; ((x,y) \in R \Rightarrow (x,z) \in S) \},
\]
and similarly for $S \SLU R$. This definition of division operations significantly depends on the choice of $U$. In particular, replacing $U$ by a superset might alter $\BSU$ and $\SLU$. Square R-models are a particular case of relativised ones, with $U = W \times W$.

Since $U$ is not required to be reflexive, in relativised R-models we deal with residuated semigroups, not monoids, that is, a class of models for $\LL$, not $\LE$. And indeed, as proved also by Andr\'{e}ka and Mikul\'{a}s~\cite{AndrekaMikulas1994}, $\LL$ is complete w.r.t.\ relativised R-models. Moreover, unlike the $\LE$ case, this result keeps valid for the extension of $\LL$ with $\wedge$.

Before going further, let us briefly compare R-models with other classes of models for the Lambek calculus and its extensions, which fit into the general algebraic framework. 

The original linguistic motivation of the Lambek calculus suggests interpretations on the algebra of formal languages. Such models are called {\em language models} or {\em L-models.} 
In L-models, multiplication is pairwise concatenation, and divisions are defined in a natural way. The difference between $\LL$ and $\LE$ is reflected by absence or presence of the empty word in the languages considered. Refraining from the empty word (which is, by the way, motivated by linguistic applications~\cite[\S\,2.5]{MootRetore}) modifies the definition of divisions, just like relativisation in R-models. Completeness theorems for $\LL$ and $\LE$ w.r.t.\ corresponding versions of L-models were proved by Pentus~\cite{Pentus1995,PentusFmonov}.
Language models might look quite similar to relational ones. However, there is the following significant difference: while any two words can be concatenated, for arrows in relations (i.e., pairs of elements of $W$) this is possible only if the end of the first one coincides with the beginning of the second one. Thus, unlike algebras of formal languages, algebras of relations have {\em divisors of zero:} if $|W| > 1$, there exist non-empty relations $R$ and $S$ over $W$ such that $R \circ S = \varnothing$.

Language models may be viewed as models on powersets of free semigroups (for $\LL$) or free monoids (for $\LE$). This suggests considering models on powersets of arbitrary semigroups or monoids (see~\cite{Buszkowski1986,Buszkowski1996}), since all of them have well-defined division operations, and also zero, unit, and intersection. 

The well-known {\em phase semantics} for linear logic~\cite{Girard1995,Abrusci1991JSL,deGroote2005,Kanovich2006}, which is also connected to denotational semantics~\cite{Bucciarelli2001}, is another species of semantics based on powersets of monoids. The crucial difference from L-models, however, is the usage of a {\em closure operator,} which in case of phase semantics is the operator taking the biorthogonal of a subset. Another variation of L-models which also uses a closure operation, but of a different nature, is the {\em syntactic concept lattice} semantics proposed by Wurm~\cite{Wurm2017}. In the same article Wurm also proposed a similar variant of R-models, also augmented with closure operations. 

In the presence of closure operations,
completeness proofs run more smoothly and cover richer extensions of the original system. In particular, besides meet one can also consider join (additive disjunction). For relational or language semantics without closure operations, where meet and join are set-theoretic intersection and union, adding join immediately leads to incompleteness, since the distributivity law (and its corollaries with only divisions and join, but not meet) is generally true, but not derivable in the calculus~\cite{OnoKomori1985,KanKuzSce2021IC}. With closure operators, distributivity is no longer generally true, which opens the way to completeness in the presence of join. In this article, we consider R-models without closure operators, and therefore consider only meet, not join.

Finally, relational models considered in this article should not be confused with {\em ternary} relational semantics for substructural logics (in particular, variants of the Lambek calculus), which offer a more Kripke-style approach (see~\cite{Dosen1992,Coumans2014}).

The rest of the article is organised as follows.
\begin{itemize}
\item
 In Section~\ref{S:calculus}, we give accurate definitions of the calculus in question, denoted by $\LEwuz$, and R-models for it. Here we also formulate known completeness results, by Andr\'eka and Mikul\'as~\cite{AndrekaMikulas1994,Mikulas2015Synthese,Mikulas2015Studia}. 
\item In Section~\ref{S:constants}, we provide examples showing that each of the explicit constants $\Zero$ and $\One$ ruins completeness, even in the weak sense (the example for $\One$ is not a new one). As a workaround, we relax the definition of R-models and define {\em R-models with non-standard interpretations of constants.} 
\item For these non-standard models, in Section~\ref{S:weak} we prove weak completeness of $\LEwuz$. As a corollary, we obtain an alternative proof of Mikul\'as' completeness result~\cite{Mikulas2015Synthese,Mikulas2015Studia} for the fragment of $\LEwuz$ without constants. 
\item Strong completeness fails, as proved in Section~\ref{S:counterexample}. The counterexample used in this section was suggested by Mikul\'as~\cite{Mikulas2015Studia}, but without proof. Here we fill this gap.
\item In Section~\ref{S:iterative}, we extend our weak completeness result to an infinitary setting. Namely, we consider so-called {\em iterative divisions,} which are Kleene stars in denominators of division operations. Being axiomatised by $\omega$-rules, iterative divisions are in fact infinite intersections. The result of Section~\ref{S:iterative} solves a natural question suggested by the earlier work~\cite{KuznetsovRyzhkova2020}. Namely, while~\cite{KuznetsovRyzhkova2020} features R-completeness for the calculus with positive iterative divisions under Lambek's restriction, Theorem~\ref{Th:itdiv} of Section~\ref{S:iterative} provides the corresponding completeness result without Lambek's restriction. 
\item The last two sections are devoted to the product-free fragment of~$\LEwuz$. We show (Section~\ref{S:strong}) that, if one removes the product, then strong completeness can be restored, but w.r.t.\ a modified class of models. However, inside the proof we essentially use the product, which raises a conservativity issue. Conservativity is proved in Section~\ref{S:exponential} in the strong form, i.e., for derivability from hypotheses, using an extension of $\LEwuz$ with the exponential modality (which allows a modalised variant of deduction theorem, cf.~\cite{LMSS}). An important corollary of the strong completeness result is that if we consider the product-free fragment without constants (that is, in the language of $\BS$, $\SL$, and $\wedge$), then it will be strongly complete w.r.t.\ square R-models in the standard sense.
\end{itemize}

This journal article is an extended version of the conference paper~\cite{RAMICS}, presented at RAMiCS 2021 and published in its proceedings. The extension here is twofold. First, along with the unit constant~$\One$, now we also consider the zero constant~$\Zero$. We show that this constant also raises issues with completeness, and resolve them using a non-standard interpretation for~$\Zero$ also. Second, we add Section~\ref{S:strong} (and Section~\ref{S:exponential} supporting it) with the strong completeness result for the product-free fragment. This result was presented as a short talk at AiML 2022, without formal publication.

\section{The Lambek Calculus with Intersection and R-Models}\label{S:calculus}

Let us recall the formulation of $\LEwuz$, the Lambek calculus with intersection and constants, in the form of a Gentzen-style sequent calculus. Formulae are constructed from variables ($p,q,r,\ldots$) and constants $\Zero$ (zero) and $\One$ (unit) using four binary connectives: $\cdot$ (multiplication), $\BS$ (left division), $\SL$ (right division), and $\wedge$ (intersection). The set of all formulae is denoted by $\Fm$. Formulae are denoted by capital Latin letters. Capital Greek letters denote sequences of formulae.
Sequents are expressions of the form $\Pi \yields B$. (Due to the non-commutative nature of the Lambek calculus, order in $\Pi$ matters.) Here $\Pi$ is called the antecedent and $B$ the succedent of the sequent.

Letter $\Lambda$, which appears in the name of the calculus, denotes the empty sequence of formulae. However, in sequents with empty antecedents we omit it and write just ${\yields B}$ instead of $\Lambda \yields B$.

The axioms and inference rules are as follows:
\[
 \infer[Id]{A \yields A}{}
 \qquad
 \infer[Cut]
 {\Gamma, \Pi, \Delta \yields C}
 {\Pi \yields A & \Gamma, A, \Delta \yields C}
\]
\[
 \infer[\BS L]
 {\Gamma, \Pi, A \BS B, \Delta \yields C}
 {\Pi \yields A & \Gamma, B, \Delta \yields C}
 \qquad
 \infer[\BS R]
 {\Pi \yields A \BS B}
 {A, \Pi \yields B}
 \qquad
 \infer[\cdot L]
 {\Gamma, A \cdot B, \Delta \yields C}
 {\Gamma, A, B, \Delta \yields C}
\]
\[
 \infer[\SL L]
 {\Gamma, B \SL A, \Pi, \Delta \yields C}
 {\Pi \yields A & \Gamma, B, \Delta \yields C}
 \qquad
 \infer[\SL R]
 {\Pi \yields B \SL A}
 {\Pi, A \yields B}
 \qquad
 \infer[\cdot R]
 {\Pi, \Delta \yields A \cdot B}
 {\Pi \yields A & \Delta \yields B}
\]
\[
 \infer[\wedge L_1]
 {\Gamma, A \wedge B, \Delta \yields C}
 {\Gamma, A, \Delta \yields C}
 \qquad
 \infer[\wedge L_2]
 {\Gamma, A \wedge B, \Delta \yields C}
 {\Gamma, B, \Delta \yields C}
 \qquad
 \infer[\wedge R]
 {\Pi \yields A \wedge B}
 {\Pi \yields A & \Pi \yields B}
\]
\[
\infer[\Zero L]{\Gamma, \Zero, \Delta \yields C}{}
\qquad
\infer[\One L]{\Gamma, \One, \Delta \yields C}{\Gamma, \Delta \yields C}
\qquad
\infer[\One R]{{} \yields \One}{}
\]

We consider both {\em pure derivability} of sequents and {\em derivability from hypotheses.} For the latter, having a set of sequents $\Hc$, we add them as extra non-logical axioms, and derive a given sequent $\Pi \to B$ using both the new axioms and axioms and inference rules of the original calculus. In this case, we say that $\Pi \to B$ is derivable, or {\em syntactically follows} from $\Hc$.
Pure derivability is derivability from an empty $\Hc$, i.e., using only axioms and rules of the original calculus.

For pure derivability, $\LEwuz$ enjoys cut elimination, that is, any derivable sequent can be derived without using $Cut$. The proof of cut elimination is standard, being an easy extension of Lambek's argument~\cite{Lambek1958} for the Lambek calculus. For derivability from hypotheses, cut elimination, in general, does not hold.

Besides $\LEwuz$ itself, we consider its {\em elementary fragments} with restricted sets of connectives and constants. Such a fragment is obtained by merely removing all axioms and rules which operate the connectives and constants which do not belong to the restricted set. In particular, we shall use the notation $\LEw$ for the fragment of $\LEwuz$ without constants $\Zero$ and $\One$ and $\LE$ for the fragment with only $\BS$, $\SL$, and $\cdot$.

Such a definition of elementary fragment raises issues of {\em conservativity.} For pure derivability, these issues are resolved by cut elimination. All rules, except $Cut$, enjoy the subformula property, therefore the set of theorems of an elementary fragment coincides with the set of sequents in the corresponding restricted language which are derivable in the full system $\LEwuz$. For derivability from hypotheses, the issue is more subtle. In principle, a derivation of $\Pi \to B$ from $\Hc$ which uses $Cut$ could involve connectives or constants which do not belong to the restricted language of $\Hc$ and $\Pi \to B$. We address this question below in Section~\ref{S:exponential}.

The original Lambek calculus $\LL$ is obtained from $\LE$ by imposing the so-called {\em Lambek's non-emptiness restriction,} that is, requiring non-emptiness of $\Pi$ in the $\BS R$ and $\SL R$ rules. This condition ensures that antecedents of all sequents in $\LL$ are non-empty ($\BS R$ and $\SL R$ are the only two rules which could possibly produce an empty antecedent). As noticed in the introduction, Lambek's restriction has linguistic motivation~\cite[\S\,2.5]{MootRetore}. One can also consider a version of $\LEw$ with Lambek's restriction, denoted by $\LLw$. On the other hand, the constant~$\One$ is incompatible with Lambek's restriction, so we shall not add constants to $\LL$ or~$\LLw$.

It is important to keep in mind that $\LL$ is not a conservative fragment of $\LE$, neither $\LLw$ is such of $\LEw$. Even if the sequent has a non-empty antecedent, empty antecedents could be necessary inside its derivation. As noticed in the introduction, an example is $(p \BS p) \BS q \yields q$, which is derivable in $\LE$, but not in $\LL$. Therefore, there is no easy way of translating results between $\LL$ and $\LE$, and certain properties of these systems differ, as we shall see below. 

 Now let us define the {\em relational semantics.} Here we define it for calculi without constants $\Zero$ and $\One$. Those are handled in the next section, since the most natural ways to interpret constants lead to incompleteness, and therefore non-standard models will be introduced.
 
 \begin{defi}\label{Df:Rmod}
 A relativised relational model (R-model)  is a triple $\Md = (W, U, v)$, where $W$ is a non-empty set, $U \subseteq W \times W$ is a transitive relation on $W$ called the  universal one, and $v \colon \Fm \to \Pc(U)$ is a valuation function mapping formulae to subrelations of $U$. The valuation function should obey the following conditions:
 \allowdisplaybreaks
 \begin{align*}
  & v(A \cdot B) = v(A) \circ v(B) = \{ 
  (x,z) \mid \exists y \in W \: \bigl( (x,y) \in v(A) \mbox{ and } (y,z) \in v(B) \bigr) \};\\
  & v(A \BS B) = v(A) \BSU v(B) = \{ 
  (y,z) \in U \mid \forall x  \in W \:
  \bigl( (x,y) \in v(A) \Rightarrow (x,z) \in v(B)  \bigr) \};\\
  & v(B \SL A) = v(B) \SLU v(A) = \{
  (x,y) \in U \mid \forall z  \in W \: \bigl(
  (y,z) \in v(A) \Rightarrow (x,z) \in v(B)  \bigr) \};\\
  & v(A \wedge B) = v(A) \cap v(B) = \{ (x,y) \mid (x,y) \in v(A) \mbox{ and } (x,y) \in v(B) \}.
 \end{align*}
 \end{defi}
 
 \begin{defi}
 An R-model $\Md = (W,U,v)$ is a square one if $U = W \times W$.
 \end{defi}

Arbitrary R-models and square R-models form natural classes of models for $\LLw$ and $\LEw$ respectively (and, thus, for $\LL$ and $\LE$). Let us define the truth condition of sequents in R-models.

\begin{defi}\label{Df:truth}
 A sequent of the form $A_1, \ldots, A_n \yields B$ is true in model $\Md = (W,U,v)$, if 
 $v(A_1) \circ \ldots \circ v(A_n) \subseteq v(B)$. For sequents with empty antecedents, truth is defined only in square R-models: ${\yields B}$ is true in $\Md = (W, W \times W, v)$, if $\delta = \{ (x,x) \mid x \in W \} \subseteq v(B)$.
\end{defi}

Let us  also recall the general notion of {\em strong} soundness and completeness of a logic $\Lc$ (formulated as a sequent calculus) w.r.t.\ a class of models $\Kc$.

\begin{defi}
Let $\Pi \to B$ and $\Hc$  be, respectively, a sequent and a set of sequents in the language of $\Lc$. The sequent $\Pi \to B$ semantically follows from $\Hc$ on the class of models $\Kc$, if for any model from $\Kc$ in which all sequents from $\Hc$ are true, the sequent $\Pi \to B$ is also true. This is denoted by $\Hc \vDash_\Kc \Pi \to B$.
\end{defi}

\begin{defi}
In the notations of the previous definition, $\Pi \to B$ syntactically follows from $\Hc$ in the logic $\Lc$, if $\Pi \to B$ is derivable in the calculus for $\Lc$ extended with sequents from $\Hc$ as extra axioms. This is denoted by $\Hc \vdash_\Lc \Pi \to B$.
\end{defi}

\begin{defi}
The logic $\Lc$ is strongly sound w.r.t.\ the class of models $\Kc$, if $\Hc \vdash_\Lc \Pi \to B$ entails $\Hc \vDash_\Kc \Pi \to B$ for any $\Pi \to B$ and $\Hc$.
\end{defi}

\begin{defi}
The logic $\Lc$ is strongly complete w.r.t.\ the class of models $\Kc$, if $\Hc \vDash_\Kc \Pi \to B$ entails $\Hc \vdash_\Lc \Pi \to B$ for any $\Pi \to B$ and $\Hc$.
\end{defi}

The Lambek calculus is a substructural system which does not enjoy a deduction theorem. Therefore, strong soundness and completeness, even for finite sets $\Hc$, are significantly different from their more usual weak counterparts (that derivability of a sequent without hypotheses yields its truth in all models from the given class, and the other way round). 

One can easily check that $\LLw$ and $\LEw$ are strongly sound w.r.t.\ the corresponding class of R-models: namely, all R-models for $\LL$ and square ones for $\LE$. The situation with completeness is non-trivial.  Andr\'{e}ka and Mikul\'{a}s~\cite{AndrekaMikulas1994} proved the following strong completeness results:

\begin{thm}[Andr\'eka, Mikul\'as 1994]\label{Th:AndrekaMikulas}
The calculus $\LLw$ is strongly complete w.r.t.\ the class of all R-models.
\end{thm}

\begin{thm}[Andr\'eka, Mikul\'as 1994]\label{Th:AndrekaMikulasLE}
The calculus $\LE$ is strongly complete w.r.t.\ the class of square R-models.
\end{thm}

The arguments used for proving these two theorems, being similar, are yet not completely identical. In particular,  Theorem~\ref{Th:AndrekaMikulas} is valid for the calculus with intersection, while the proof of Theorem~\ref{Th:AndrekaMikulasLE} cannot be easily extended to such a calculus, i.e., to $\LEw$.

Later on, however, Mikul\'as~\cite{Mikulas2015Synthese,Mikulas2015Studia} managed to modify the proof of Theorem~\ref{Th:AndrekaMikulasLE} for $\LEw$, but this modification establishes only weak completeness:
\begin{thm}[Mikul\'as 2015]\!\footnote{Here ``Mikul\'{a}s 2015'' refers both to~\cite{Mikulas2015Synthese} and~\cite{Mikulas2015Studia}, which feature different proofs of Theorem~\ref{Th:Mikulas}.}
\label{Th:Mikulas}
 If a sequent (in the language of $\cdot,\BS,\SL,\wedge$) is true in all square R-models, then it is derivable in $\LEw$.
\end{thm}

In Section~\ref{S:counterexample} we shall show that strong completeness for $\LEw$ fails.

\section{Constants Zero and One: Nonstandard Models}\label{S:constants}

The situation with constants~$\Zero$ and~$\One$ is tricky. Naturally, $\One$ should be interpreted as the neutral element for multiplication (since $A \cdot \One \leftrightarrow A \leftrightarrow \One \cdot A$) and $\Zero$ as the smallest element in the preorder (due to derivability of $\Zero \to A$). This suggests the following interpretation of constants in R-models, which we call the {\em standard interpretation:}
\begin{align*}
 & v(\Zero) = \varnothing; \\
 & v(\One) = \delta = \{ (x,x) \mid x \in W \}.
\end{align*}

Unfortunately, under the standard interpretation completeness fails, even in the weak sense, for each of the constants. Let us start with constant~$\Zero$.

\begin{prop}
 The sequent 
 \[
  \Zero \SL (\Zero \SL p), \Zero \SL (\Zero \SL q) \yields (\Zero \SL (\Zero \SL q)) \cdot (\Zero \SL (\Zero \SL p))
 \]
is true in all square R-models (under the standard interpretation of $\Zero$), but not derivable in $\LEwuz$.
\end{prop}

\begin{proof}
 Let us first check that sequent is generally true on R-models. For an arbitrary relation $R$ we have:
 \begin{multline*}
 \varnothing \SL R = \{ (x,y) \in W \times W \mid \forall z \in W \: \bigl( (y,z)  \in R \Rightarrow (x,z) \in \varnothing \bigr) \} =\\
 \{ (x,y) \in W \times W \mid \forall z \in W \: (y,z) \notin R \}.
 \end{multline*}
Further,
\begin{multline*}
\varnothing \SL (\varnothing \SL R) = \{
(x,y) \in W \times W \mid \forall z \in W \: (y,z) \notin \varnothing \SL R \} = \\
\{ (x,y) \in W \times W \mid \forall z \in W \: \exists w \in W \: (z,w) \in R \}.
\end{multline*}
The condition on $(x,y)$ here does not depend on $(x,y)$ itself. Hence,  $\varnothing \SL (\varnothing \SL R)$ is either $\varnothing$ or $W \times W$, for any relation $R$. For these two relations, composition is commutative:
\[
 \varnothing \circ (W \times W) = \varnothing = (W \times W) \circ \varnothing.
\]
This yields $(\varnothing \SL (\varnothing \SL P)) \circ (\varnothing \SL (\varnothing \SL Q)) \subseteq (\varnothing \SL (\varnothing \SL Q)) \circ (\varnothing \SL (\varnothing \SL P))$ for any $P,Q$. Taking $P = v(p)$, $Q = v(q)$, and recalling that $v(\Zero) = \varnothing$ (standard interpretation), we conclude that the desired sequent is true.

Now let us perform cut-free proof search in order to show that this sequent is not derivable in $\LEwuz$. Two possible cases for the lowermost rule are $\SL L$ and $\cdot R$. In the first case, the left premise of this rule should be $\Zero \SL (\Zero \SL q) \yields \Zero \SL p$ (since ${} \yields \Zero \SL p$ and ${} \yields \Zero \SL q$ are obviously not derivable). The only non-trivial attempt to derive this sequent is as follows:
\[
 \infer
 {\Zero \SL (\Zero \SL q) \yields \Zero \SL p}
 {\infer{\Zero \SL (\Zero \SL q), p \yields \Zero}{\infer{p \yields \Zero \SL q}{p,q \yields \Zero} & \Zero \yields \Zero}}
\]
and it fails.

The second possibility is $\cdot R$. One can easily see that ${\yields \Zero \SL (\Zero \SL q)}$ is not derivable, neither is ${\yields \Zero \SL (\Zero \SL p)}$. Hence, the only remaining attempt is as follows:
\[
 \infer
 {\Zero \SL (\Zero \SL p), \Zero \SL (\Zero \SL q) \yields (\Zero \SL (\Zero \SL q)) \cdot (\Zero \SL (\Zero \SL p))}
 {\infer{\Zero \SL (\Zero \SL p) \yields \Zero \SL (\Zero \SL q)}{\infer{\Zero \SL (\Zero \SL p), \Zero \SL q \yields \Zero}
 {\infer{\Zero \SL q \yields \Zero \SL p}{\Zero \SL q, p \yields \Zero} & \Zero \yields \Zero}} & \Zero \SL (\Zero \SL q) \yields \Zero \SL (\Zero \SL p)}
\]
The top sequent $\Zero \SL q, p \yields \Zero$ is not derivable.
\end{proof}

Next, we recall two previously known counterexamples to weak completeness with constant~$\One$, in square R-models with standard interpretation of this constant.

One example, 
\[
 \One \wedge p \wedge q \yields (\One \wedge p) \cdot (\One \wedge q),
\]
was given by Andr\'eka and Mikul\'as~\cite{AndrekaMikulas2011}. In the particular case of $p = q$, this yields the contraction (``doubling'') principle for formulae of the form $\One \wedge A$, that is, $\One \wedge A \to (\One \wedge A) \cdot (\One \wedge A)$. Another such example was given by Buszkowski~\cite{BuszkoRelMiCS}: 
\[
 \One \SL (p \SL p) \yields (\One \SL (p \SL p)) \cdot (\One \SL (p \SL p)).
\]
This
example uses division instead of intersection, and again it is a form of contraction.

Thus, constructing a complete axiomatisation for the standard interpretation of constants $\Zero$ and $\One$ in R-models, even if this is possible, is a non-trivial open question.

We overcome this issue by extending the class of models being considered,
thus restoring completeness for the original system $\LEwuz$. The idea is as follows:
while $\delta$ is the only neutral element for the set of {\em all} binary relations on $W$, for
a set which includes only some relations (and does not include $\delta$), the neutral
element could be a different relation. Similarly for zero, a designated subclass of binary relations could have the smallest element which is not $\varnothing$. This leads to the following definition.

\begin{defi}\label{Df:Aunit}
 Let $\Af \subseteq \Pc(W \times W)$ be a family of binary relations over $W$, closed under $\circ$, $\BS$, $\SL$, and $\cap$. Relation $\UA \in \Af$ is called the $\Af$-unit if $\UA \circ R = R \circ \UA = R$ for any $R \in \Af$. Relation $\ZA \in \Af$ is called the $\Af$-zero if $\ZA \subseteq R$ for any $R \in \Af$.
\end{defi}

A standard algebraic argument shows that the $\Af$-unit, if it exists, is unique. Indeed, for another $\Af$-unit $\UA' \in \Af$ we have $\UA' = \UA' \circ \UA = \UA$. For the $\Af$-zero, uniqueness is obvious.

However, $\UA$ is not necessarily the diagonal relation $\delta = \{ (x,x) \mid x \in W \}$, as the latter may be outside $\Af$. For example, let $W$ be a non-empty set and let $W' = W \times \{ 1,2 \}$. For each relation $R$ on $W$ let us define a relation $R'$ as follows: $(x,i) R' (y,j)$, if $xRy$ and $i \le j$. Let $\Af$ be the class of relations of the form $R'$. Then $\UA = \delta' = \{ ((x,i),(x,j)) \mid i \le j \}$. Similarly, $\ZA$ is not necessarily $\varnothing$; in fact, it is just the intersection of all relations from $\Af$.

Let us prove two properties of $\ZA$ and $\UA$.

\begin{lem}\label{Lm:zerozero}
 If $\ZA$ is the $\Af$-zero, then it is the zero for relation composition, i.e., $\ZA \circ R = R \circ \ZA = \ZA$ for any $R \in \Af$.
\end{lem}

\begin{proof}
 Let us show $\ZA \circ R = \ZA$, the other equality is symmetric. First, since $\Af$ is closed under composition, we have $\ZA \circ R \in \Af$, and therefore $\ZA \subseteq \ZA \circ R$. For the opposite inclusion, we notice that $\ZA \SL R$ is also in $\Af$, whence $\ZA \subseteq \ZA \SL R$, which is equivalent to $\ZA \circ R \subseteq \ZA$.
\end{proof}

\begin{rem}
 Lemma~\ref{Lm:zerozero} is actually a general algebraic fact that, in residuated partially ordered structures, the least object of the partial order should always be the zero for the multiplication operation. Division operations are crucial here.
\end{rem}

\begin{lem}\label{Lm:deltaunit}
 If $R \in \Af$ and $\UA$ is the $\Af$-unit, then $\UA \subseteq R$ if and only if $\delta \subseteq R$. In particular, $\delta \subseteq \UA$. 
\end{lem}

\begin{proof}
 Let us first show that for any $R \in \Af$ we have $R = R \SL \UA$. One inclusion is easy: $R \circ \UA \subseteq R$ yields $R \subseteq R \SL \UA$. For the other inclusion, we first notice that $(R \SL \UA) \circ \UA \subseteq R$ (this follows from $R \SL \UA \subseteq R \SL \UA$). Now, since $(R \SL \UA) \in \Af$, we have $R \SL \UA = (R \SL \UA) \circ \UA \subseteq R$.
 Using $R = R \SL \UA$, we build a chain of equivalences:
 $\UA \subseteq R \iff \delta \circ \UA \subseteq R \iff \delta \subseteq R \SL \UA \iff \delta \subseteq R.
 $
\end{proof}

Due to this lemma, we may keep the truth definition for sequents with empty antecedents the same. That is, we do not need to replace $\delta$ with $\UA$.

\begin{rem}\label{Rem:outside}
 It is important to note that computations in the proof of Lemma~\ref{Lm:deltaunit} go {\em beyond} the family $\Af$, since it is possible that $\delta \notin \Af$. This does not cause any problems, because the introduction of $\Af$ does not change the definitions of operations, in particular, divisions. They are still computed on the whole $\Pc(W \times W)$, as in a standard square R-model.
\end{rem}

\begin{rem}
 One of the reviewers noticed that $\Af$-units may be viewed as {\em ``super-identities,''} as opposed to more well-known {\em subidentities,} which appear, e.g., in the theory of domain semirings~\cite{DesharnaisStruth}. The behaviour of such ``super-identities'' in  more general algebraic settings might be a topic for further study.
\end{rem}

Now we are ready to define our non-standard models.

\begin{defi}\label{Df:RmodNS}
 A non-standard square R-model is a structure $\Md^{\Af} = (W, \Af, \UA, \ZA, v)$, where $W$ is a non-empty set; $\Af \subseteq \Pc(W \times W)$ is a family of binary relations on $W$, closed under $\circ$, $\BS$, $\SL$, and $\cap$; $\UA$ is the $\Af$-unit; $\ZA$ is the $\Af$-zero; $v \colon \Fm \to \Af$ is a valuation function mapping formulae to relations from the family $\Af$. The valuation function should obey the conditions from Definition~\ref{Df:Rmod}, with $U = W \times W$, and, additionally, $v(\One) = \UA$ and $v(\Zero) = \ZA$. The truth of a sequent in a non-standard square R-model is defined exactly as in Definition~\ref{Df:truth}.
\end{defi}

\begin{prop}\label{Prop:sound}
The calculus $\LEwuz$ is strongly sound w.r.t.\ the class of non-standard square R-models.
\end{prop}

\begin{proof}
As usual, we proceed by induction on the derivation. The interesting cases are $\One L$,  $\One R$,  and $\Zero L$, as others are copied from the standard strong soundness proof of $\LEw$ w.r.t.\ square R-models (without constants~$\Zero$ and~$\One$). 

For $\One R$, we have to show that ${\to \One}$ is true, that is, $\delta \subseteq v(\One) = \UA$. This is a particular case of Lemma~\ref{Lm:deltaunit}.

For $\One L$, we consider two cases. If both $\Gamma$ and $\Delta$ are empty, then our induction hypothesis gives $\delta \subseteq v(C)$. By Lemma~\ref{Lm:deltaunit}, this is equivalent to $\UA \subseteq v(C)$ (recall that $v(C) \in \Af$), which is the truth of $\One \to C$. If, say, $\Delta$ is non-empty, then let $D_1$ be the first formula of $\Delta$. By definition of the $\Af$-unit, we have $v(\One) \circ v(D_1) = \UA \circ v(D_1) = v(D_1)$. Thus, interpretations of left-hand sides of the premise and the conclusion are identical. The case of non-empty $\Gamma$ is symmetric. 

For $\Zero L$, since the antecedent includes $\Zero$, its interpretation is $\ZA$ by Lemma~\ref{Lm:zerozero}. By definition of the $\Af$-zero, $\ZA \subseteq v(C)$, where $C$ is the succedent.
\end{proof}

As for completeness, we prove only its weak version (Section~\ref{S:weak}). For strong completeness, there is a counterexample (Section~\ref{S:counterexample}).

\section{Weak Completeness}\label{S:weak}

We prove weak completeness of $\LEwuz$ w.r.t.\ the class of non-standard square R-models, as defined in the previous section.

\begin{thm}\label{Th:weak}
If a sequent (in the language of $\BS,\SL,\cdot,\wedge,\One,\Zero$) is true in all non-standard square R-models, then it is derivable in $\LEwuz$.
\end{thm}

Our proof follows the line of the proof of Theorem~\ref{Th:AndrekaMikulas} (surprisingly, not Theorem~\ref{Th:AndrekaMikulasLE}, see Remark~\ref{Rem:filter}  below): we build a labelled graph with specific properties and use it to construct a universal model.

Throughout this section, $\vdash A \to B$ means ``$A \to B$ is derivable in $\LEwuz$.''

\begin{lem}\label{Lm:graph}
 There exists a labelled directed graph $G = (V, E, \ell)$, where $V \ne \varnothing$, $E \subseteq V \times V$, and $\ell \colon E \to \Fm$, such that the following holds:
 \begin{enumerate}
  \item\label{It:transrefl} $E$ is transitive;
  \item\label{It:unit} $E$ is reflexive and $\ell(x,x) = \One$ for any $x \in V$;
  \item\label{It:antisymm} $E$ is antisymmetric: if $x \ne y$ and $(x,y) \in E$, then $(y,x) \notin E$;
  \item\label{It:cdot1} if $(x,y) \in E$ and $(y,z) \in E$, then $\vdash \ell(x,z) \yields \ell(x,y) \cdot \ell(y,z)$;
  \item\label{It:cdot2} if\/ $\vdash \ell(x,z) \yields B \cdot C$, then
  there exists  $y \in V$ such that $(x,y) \in E$, $(y,z) \in E$,
  $\vdash \ell(x,y) \yields B$, and $\vdash \ell(y,z) \yields C$;
  \item\label{It:Aprev} for any $y \in V$ and any formula $A$ there exists $x \in V$ such that  
  for any $z \in V$, if $(y,z) \in E$, then $\ell(x,z) =  A \cdot \ell(y,z)$; 
  \item\label{It:Asucc} for any $y \in V$ and any formula $A$ there exists $z \in V$ such that for any $x \in V$, if $(x,y) \in E$, then $\ell(x,z) = \ell(x,y) \cdot A$.
 \end{enumerate}

\end{lem}

Before proving Lemma~\ref{Lm:graph}, let us use it to establish Theorem~\ref{Th:weak}.

\begin{proof}[Proof of Theorem~\ref{Th:weak}]
 Using the graph $G$ provided by Lemma~\ref{Lm:graph}, we construct a {\em universal} non-standard square R-model $\Md_0^{\Af} = (W, \Af, \UA, \ZA, v)$ in the following way:
 \begin{align*}
  & W = V; && &
  & v(A) = \{ (x,y) \in E \mid\ \vdash \ell(x,y) \yields A \}; \\
  & \Af = \{ v(A) \mid A \in \Fm \}; && &
  & \UA = v(\One);\ \ZA = v(\Zero).
 \end{align*}

 Let us prove that the valuation function $v$ commutes with the operation of our calculus. 
 
 \vskip 3pt
 {\em Multiplication.} 
  If $(x,z) \in v(B \cdot C)$, then $\vdash \ell(x,z) \yields B \cdot C$. By property~\ref{It:cdot2} of graph $G$, there exists such $y$ that $(x,y) \in v(B)$ and $(y,z) \in v(C)$. Therefore, $(x,z) \in v(B) \circ v(C)$.
  This establishes the inclusion $v(B \cdot C) \subseteq v(B) \circ v(C)$.
  
  For the opposite inclusion, take $(x,y) \in v(B)$ and $(y,z) \in v(C)$. By transitivity, $(x,z) \in E$. By property~\ref{It:cdot1} of $G$, 
  $\vdash \ell(x,z) \yields \ell(x,y) \cdot \ell(y,z)$. We derive $\ell(x,z) \yields B \cdot C$ as follows:
  \[
   \infer[Cut]
   {\ell(x,z) \yields B \cdot C}
   {\ell(x,z) \yields \ell(x,y) \cdot \ell(y,z) & \infer[\cdot L]{\ell(x,y) \cdot \ell(y,z) \yields B \cdot C}
   {\infer[\cdot R]{\ell(x,y), \ell(y,z) \yields B \cdot C}{\ell(x,y) \yields B & \ell(y,z) \yields C}}}
  \]
Therefore, $(x,z) \in v(B \cdot C)$.
  
  \vskip 3pt
  {\em Division.} Let $(y,z) \in v(A \BS B)$, that is, $\vdash \ell(y,z) \yields A \BS B$. Take an arbitrary $x \in W$ such that $(x,y) \in v(A)$, that is, $\vdash \ell(x,y) \yields A$. Now by transitivity $(x,z) \in E$, and $\ell(x,z) \yields B$ is derived using two cuts, with $\ell(x,z) \yields \ell(x,y) \cdot \ell(y,z)$ (property~\ref{It:cdot1} of $G$) and $A \cdot (A \BS B) \yields B$.
  This establishes the inclusion $v(A\BS B) \subseteq v(A) \BS v(B)$.
  
  For the opposite inclusion, take $(y,z) \in v(A) \BS v(B)$ and apply property~\ref{It:Aprev} to $y$ and $A$. 
  For the vertex $x \in W$ given by this property, we have $\ell(x,y) = A \cdot \ell(y,y) = A \cdot \One$ 
  and $\ell(x,z) = A \cdot \ell(y,z)$. The first condition gives $(x,y) \in v(A)$ (because $\vdash A \cdot \One \yields A$). Hence, $(x,z) \in v(B)$, i.e., $(x,z) \in E$ and $\vdash \ell(x,z) \yields B$. In particular, $(x,z) \in E$ gives $(y,z) \in E$ (which was not guaranteed in advance). Since  $\ell(x,z) = A \cdot \ell(y,z)$, we may proceed as follows:
  \[
   \infer[\BS R]
   {\ell(y,z) \yields A \BS B}
   {\infer[Cut]{A, \ell(y,z) \yields B}
   {\infer[\cdot R]{A, \ell(y,z) \yields A \cdot \ell(y,z)}{A \yields A & \ell(y,z) \yields \ell(y,z)} & A \cdot \ell(y,z) \yields B}}
  \]
  This establishes $(y,z) \in v(A \BS B)$. Thus, we get $v(A) \BS v(B) = v(A \BS B)$.
  
The equality $v(B \SL A) = v(B) \SL v(A)$ is established symmetrically.

\vskip 3pt
{\em Intersection.} We have 
$ v(A \wedge B) = \{ (x,y) \in E \mid\ \vdash \ell(x,y) \yields A \wedge B \} = 
\{ (x,y) \in E \mid\ \vdash \ell(x,y) \yields A \mbox{ and } \vdash \ell(x,y) \yields B  \} = v(A) \cap v(B)$.
In the second equality, the $\supseteq$ inclusion is by $\wedge R$, and the $\subseteq$ one is by cut with $A \wedge B \yields A$ and $A \wedge B \yields B$.

\vskip 3pt
{\em Unit.} Here we have $v(\One) = \UA$ by definition, and $\UA$ is the $\Af$-unit.
Indeed, since 
$A$ is equivalent to $\One \cdot A$
and any relation in $\Af$ is of the form $v(A)$, we have 
\(
 \UA \circ v(A) = v(\One) \circ v(A) = 
 v(\One \cdot A) = v(A)
\). Similarly for $v(A) \circ \UA$.

\vskip 3pt
{\em Zero.} By definition $v(\Zero) = \ZA$, let us show that it is indeed the $\Af$-zero. Each relation in $\Af$ is of the form $v(A)$ for some $A \in \Fm$, so we need to show that $v(\Zero) \subseteq v(A)$ for any formula $A$. Let $(x,y) \in v(\Zero)$, that is, $\vdash \ell(x,y) \yields \Zero$. By cut with $\Zero \yields A$ (which is an instance of $\Zero L$), we get $\vdash \ell(x,y) \yields A$, therefore $(x,y) \in v(A)$.

 \vskip 5pt
 Now let us show that $\Md^\Af_0$ is indeed a universal model, that is, a sequent is true in this model if and only if it is derivable in $\LEwuz$. The interesting direction is of course the ``only if'' one (the ``if'' direction is just weak soundness).
 
 Moreover, we may consider only sequents of the form ${\yields B}$, since from $A_1, \ldots, A_n \yields B$ one can derive ${\yields A_n \BS (A_{n-1} \BS \ldots \BS (A_1 \BS B) \ldots )}$, and {\em vice versa,} and by strong soundness these two sequents are true or false in $\Md_0^\Af$ simultaneously.
 
 Let ${\yields B}$ be true in $\Md_0^{\Af}$, that is, $\delta \subseteq v(B)$. 
 Take an arbitrary $x \in V$. We have $(x,x) \in v(B)$, that is, $\vdash \ell(x,x) \yields B$. On the other hand, $\ell(x,x) = \One$ by property~\ref{It:unit} of $G$. Applying cut with ${\yields \One}$ (axiom), we derive the desired sequent ${\yields B}$.
 
 Existence of a universal model yields weak completeness: if a sequent is true in all models, then it is true in the universal one, and therefore derivable in the calculus.
\end{proof}

Now we finish our argument by proving Lemma~\ref{Lm:graph}. The spirit of this proof is the same as the central lemma of the proof of Theorem~\ref{Th:AndrekaMikulas} by Andr\'eka and Mikul\'as. However, for the step-by-step construction we use the countable schedule function, as in~\cite{KuznetsovRyzhkova2020}, which is sufficient for enumerating formulae, rather than consider abstract algebras of arbitrary cardinality, as in~\cite{AndrekaMikulas1994}; see Remark~\ref{Rem:algebras} below. The presentation of the proof closely follows the line of~\cite[Lemma~14]{KuznetsovRyzhkova2020}; the figures are adaptations of those by Andr\'eka and Mikul\'as~\cite{AndrekaMikulas1994} to the reflexive situation.

\begin{proof}[Proof of Lemma~\ref{Lm:graph}]
We construct a growing sequence of labelled graphs $G_n = (V_n, E_n, \ell_n)$, where each $G_n$ is an induced subgraph of $G_{n+1}$. The countable set of vertices $V = \bigcup_{n=0}^{\infty} V_n$ is fixed before the process starts. Our aim is the union graph $G = \left( \bigcup_{n=0}^{\infty} V_n, \bigcup_{n=0}^{\infty} E_n, \bigcup_{n=0}^{\infty} \ell_n \right)$. 

The desired properties \ref{It:transrefl}--\ref{It:cdot1} are {\em maintained} along the sequence, that is, they will hold for each $G_n$. In contrast, properties \ref{It:cdot2}--\ref{It:Asucc} are {\em achieved} only in the limit; each transition from $G_n$ to $G_{n+1}$ is a step towards satisfying one of these properties (in a particular case).

The initial graph $G_0$ is just a reflexive point with the required unit label on the loop:
$G_0 = (\{ \star \}, \{ (\star,\star) \}, (\star,\star) \mapsto \One)$. Properties~\ref{It:transrefl}--\ref{It:cdot1} are trivially satisfied.

Each step is a transition of one of three types: for $t=0,1,2$, a transition of type~$t$ is a step from $G_{3i+t}$ to $G_{3i+t+1}$. In order to ensure that all necessary transitions are eventually performed, let us define two bijective schedule functions:
\begin{align*}
& \sigma \colon \NN \to (V \times \Fm) \times \NN \\
& \varsigma \colon \NN \to (V \times V \times \Fm \times \Fm) \times \NN 
\end{align*}
Here $\sigma$ enumerates pairs of a (possible) vertex and a formula, and the second component (a natural number) ensures that each such pair is ``visited'' infinitely many times. The second function, $\varsigma$, does the same for quadruples including two vertices and two formulae.
Now let us define our transitions.

\vskip 3pt
{\em Transition of type 0,} from $G_{3i}$ to $G_{3i+1}$. Let $\sigma(i) = ((y,A),k)$. If $y  \notin V_{3i}$, we skip: $G_{3i+1} = G_{3i}$. Otherwise we add a new vertex $x \in V - V_{3i}$ (such a vertex always exists, since $V$ is countable and $V_{3i}$ is finite) 
with a loop edge $(x,x)$, $\ell(x,x) = \One$, 
and for each $z \in V_{3i}$, such that  $(y,z) \in E_{3i}$, an edge $(x,z)$ with $\ell(x,z) = A \cdot \ell(y,z)$. (In particular, we add an edge $(x,y)$ with label $A \cdot \One$.)

Let us show that properties~\ref{It:transrefl}--\ref{It:cdot1} keep valid for $G_{3i+1}$. Indeed, the new vertex $x$ is reflexive, and the loop has the correct label $\One$. Antisymmetry is also maintained: the new vertex $x$ has no ingoing edges, except the loop. 

Transitivity and property~\ref{It:cdot1} are checked as follows. We have to verify that for any $x',y',z' \in V_{3i+1}$ if $(x',y') \in E_{3i+1}$ and $(y',z') \in E_{3i+1}$, then $(x',z') \in E_{3i+1}$ and $\vdash \ell(x',z') \to \ell(x',y') \cdot \ell(y',z')$. The interesting case is when at least one of these vertices is new (that is, not from $V_{3i}$). This means $x' = x$. If $y' = x$, we trivially get $(x',z') = (y',z') \in E_{3i+1}$, and $\vdash \ell(x',z') \to \ell(x',y') \cdot \ell(y',z')$, since $\ell(x',y') = \ell(x,x) = \One$. Now let $y',z'$ be old vertices (from $V_{3i}$). Since $(x,y')$ and $(x,z')$ were added, edges $(y,y')$ and $(z,z')$ are in $E_{3i}$. Let $\ell(y,y') = B$, $\ell(y',z') = C$, and $\ell(y,z') = D$ (the latter edge exists by transitivity of $G_{3i}$). Then the picture is as follows (new edges are dashed):
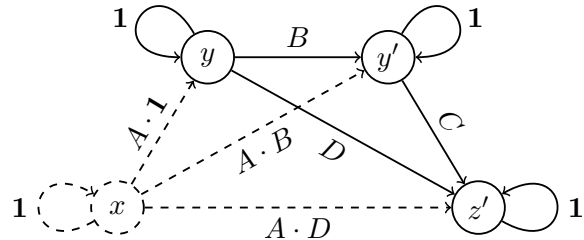
\begin{center}
\begin{tikzpicture}[line width=.7pt]
\begin{scope}[scale=.8]
 \node[circle,inner sep=2pt,draw,minimum size=7mm,dashed] (x) at (0,0) {$x$};
 \node[circle,inner sep=2pt,draw,minimum size=7mm] (y) at (1.5,2.5) {$y$};
 \node[circle,inner sep=2pt,draw,minimum size=7mm] (yx) at (4.5,2.5) {$y'$};
 \node[circle,inner sep=2pt,draw,minimum size=7mm] (zx) at (6,0) {$z'$};
 
 \draw (x) edge[->,dashed] node [midway,above,sloped] {$A \cdot \One$}  (y) ;
 
 \draw (y) edge[->] node[midway,above] {$B$} (yx);
 
 \draw (yx) edge[->] node[midway,above,sloped] {$C$} (zx);
 
 \draw (x) edge[->,dashed] node[midway,below] {$A \cdot D$} (zx);
 
 \draw (y) edge[->] node[midway,below,sloped] {$D$} (zx);
 
 \draw (x) edge[->,dashed] node[midway,below,sloped] {$A \cdot B$} (yx);
 
 \draw (x) edge[in=150,out=210,looseness=8,->,dashed] node[midway,left] {$\One$} (x);
 
 \draw (zx) edge[in=30,out=-30,looseness=8,->] node[midway,right] {$\One$} (zx);
 
 \draw (y) edge[in=180,out=120,looseness=8,->] node[midway,left] {$\One$} (y);
 \draw (yx) edge[out=60,in=0,looseness=8,->] node[midway,right] {$\One$} (yx);
\end{scope}
\end{tikzpicture}
\end{center}
(Notice that some of the vertices $y$, $y'$, $z'$ could coincide; in this case, the corresponding edges are loops with label $\One$.) 

We have indeed added the necessary edge $(x,z')$, and it remains to check that $\vdash \ell(x,z') \to \ell(x,y') \cdot \ell(y',z')$, that is, $A \cdot D \to (A \cdot B) \cdot C$. By associativity, we may replace $(A \cdot B) \cdot C$ with $A \cdot (B \cdot C)$, and then $A \cdot D \to A \cdot (B \cdot C)$ is derived from $D \to B \cdot C$ by applying $\cdot R$ and $\cdot L$. The sequent $D \to B \cdot C$ is derivable by property~\ref{It:cdot1} of the old graph $G_{3i}$.

\vskip 3pt 
{\em Transition of type 1,} from $G_{3i+1}$ to $G_{3i+2}$, is similar. Let $\sigma(i) = ((y,A),k)$. If $y \notin V_{3i}$, we skip, and otherwise add a new vertex $z$ with its loop and for each $x \in V_{3i}$, if $(x,y) \in E_{3i}$, add an edge $(x,z)$ with $\ell(x,z) = \ell(x,y) \cdot A$. As for type~0, properties \ref{It:transrefl}--\ref{It:cdot1} keep valid.

\vskip 3pt
{\em Transition of type 2,} from $G_{3i+2}$ to $G_{3i+3}$. Let $\varsigma(i) = ((x,z,B,C),k)$. If $x$ or $z$ is not in $V_{3i+2}$ or if $\not\vdash \ell(x,z) \to B \cdot C$, we skip. We also skip if $x=z$: in this case, we do not need to add a new vertex to satisfy property~\ref{It:cdot2}, see below. Otherwise, we add a new vertex $y$, with its loop $(y,y)$, $\ell(y,y) = \One$, and the following edges:
\begin{itemize}
 \item edge $(r,y)$, with $\ell(r,y) = \ell(r,x) \cdot B$, for each $r$ such that $(r,x) \in E_{3i+2}$;
 \item edge $(y,s)$, with $\ell(y,s) = C \cdot \ell(z,s)$, for each $s$ such that $(z,s) \in E_{3i+2}$.
\end{itemize}
(In particular, we add edges $(x,y)$ and $(y,z)$ with labels $\One \cdot B$ and $C \cdot \One$ respectively.)
The picture in this situation is as follows:
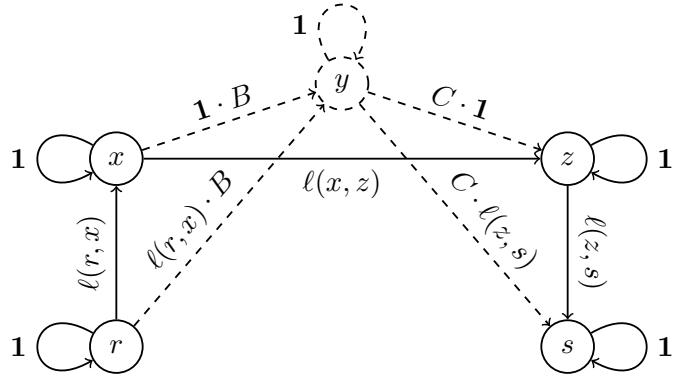
\begin{center}
\begin{tikzpicture}[line width=.7pt]
\begin{scope}[scale=1]
 \node[circle,inner sep=2pt,draw,minimum size=7mm] (x) at (0,0) {$x$};
 \node[circle,inner sep=2pt,draw,minimum size=7mm] (r) at (0,-2.5) {$r$};
 \node[circle,inner sep=2pt,draw,minimum size=7mm] (z) at (6,0) {$z$};
 \node[circle,inner sep=2pt,draw,minimum size=7mm] (s) at (6,-2.5) {$s$};
 \node[circle,inner sep=2pt,draw,minimum size=7mm,dashed] (y) at (3,1) {$y$};
 
 \draw (x) edge[->,dashed] node[midway,above,sloped] {$\One \cdot B$} (y);
 \draw (y) edge[->,dashed] node[midway,above,sloped] {$C \cdot \One$} (z);
 
 \draw (r) edge[->] node[midway,above,sloped] {$\ell(r,x)$} (x);
 \draw (z) edge[->] node[midway,above,sloped] {$\ell(z,s)$} (s);
 
 \draw (x) edge[->] node[midway,below,sloped] {$\ell(x,z)$} (z);
 
 \draw (r) edge[->,dashed] node[pos=0.4,above,sloped] {$\ell(r,x)\cdot B$} (y);
 \draw (y) edge[->,dashed] node[pos=0.6,above,sloped] {$C \cdot \ell(z,s)$} (s);
 
 \draw (x) edge[->,in=210,out=150,looseness=8] node[midway,left] {$\One$} (x);
 \draw (r) edge[->,in=210,out=150,looseness=8] node[midway,left] {$\One$} (r);

 \draw (z) edge[->,in=-30,out=30,looseness=8] node[midway,right] {$\One$} (z);
 \draw (s) edge[->,in=-30,out=30,looseness=8] node[midway,right] {$\One$} (s);
 
 \draw (y) edge[->,dashed,in=60,out=120,looseness=8] node[pos=0.2,left] {$\One$} (y);
\end{scope}
\end{tikzpicture}
\end{center}
In this picture, it is possible that $r=x$, or $z=s$, or even both. In such a case, the corresponding edge is a loop with label $\One$. However, $x \ne z$ by assumption, and also $r \ne s$ (for any $r$, $s$ in question). Indeed, if $r = s$, then by transitivity we get $(z,x) \in E_{3i+2}$, which violates antisymmetry of the old graph $G_{3i+2}$. Also, the new vertex $y$ is a distinct one.

The $r \ne s$ condition yields antisymmetry of the new graph $G_{3i+3}$. Indeed, a possible violation of antisymmetry should involve the new vertex $y$, but then the other vertex should be $r$ and $s$ at the same time.

Reflexivity of the new graph, with $\One$ labels on the loops, is by construction.

Let us check transitivity and property~\ref{It:cdot1}. Take $x',y',z'$ such that edges $(x',y')$ and $(y',z')$ belong to $E_{3i+3}$. The interesting case is when $x' \ne y'$ and $y' \ne z'$ (otherwise we just add a unit), and at least one of $x',y',z'$ is the new vertex $y$. Moreover, by antisymmetry, which we have already proved, we have $x' \ne z'$.

Consider three cases.

\vskip 5pt
{\em Case 1:} $x' = y$. Denote $s_1 = y'$ and $s_2 = z'$. Since $(y,s_1) \in E_{3i+3}$ and $s_1 \ne y$, we have $(z,s_1) \in E_{3i+2}$ and $\ell(y,s_1) = C \cdot \ell(z,s_1)$. (Possibly, $s_1 = z$.) For $s_2$, since it is also not $y$, we have $(s_1,s_2) \in E_{3i+2}$, and by transitivity of $G_{3i+2}$ we get $(z,s_2) \in E_{3i+2}$. Therefore, $(y,s_2) \in E_{3i+3}$ and $\ell(y,s_2) = C \cdot \ell(z,s_2)$.
 Now by property~\ref{It:cdot1} of the old graph we have $\vdash \ell(z,s_2) \to \ell(z,s_1) \cdot \ell(s_1,s_2)$, and via $\cdot R$, $\cdot L$, and associativity we obtain $\vdash C \cdot \ell(z,s_2) \to (C \cdot \ell(z,s_1)) \cdot \ell(s_1,s_2)$. This is the necessary sequent $\ell(y,s_2) \to \ell(y,s_1) \cdot \ell(y,s_2)$.

 Notice that here antisymmetry is crucial: otherwise, we could have $z' = s_2 = y$ (i.e., $s_1$ plays both as $s$ and $r$), in which case $\ell(y,s_2)$ would be $\One$, not $C \cdot \ell(z,s_2)$.
 
 \vskip 3pt
 {\em Case 2:} $z' = y$. Considered symmetrically.
 
 \vskip 3pt
 {\em Case 3:} $y' = y$. Denote $r = x'$ and $s = z'$; they are both distinct from $y$. We have $\ell(r,y) = \ell(r,x) \cdot B$ and $\ell(y,s) = C \cdot \ell(z,s)$. By shortcutting the path $r$--$x$--$z$--$s$ using transitivity, we see that $(r,s)$ is an edge of the old graph and $\vdash \ell(r,s) \to \ell(r,x) \cdot \ell(x,z) \cdot \ell(z,s)$. Now we recall that $\vdash \ell(x,z) \to B \cdot C$ by assumption and by cut and monotonicity conclude that $\vdash \ell(r,s) \to \ell(r,x) \cdot B \cdot C \cdot \ell(z,s)$. This is exactly (up to associativity) what we need: $\vdash \ell(r,s) \to \ell(r,y) \cdot \ell(y,s)$.

\vskip 5pt
Our construction shows that properties~\ref{It:transrefl}--\ref{It:cdot1} hold for each $G_n$. Therefore, they also hold for the limit graph $G$. Also notice that this graph is infinite, because, e.g., transitions of type~0 are used infinitely often. Hence, the set of vertices of $G$ is indeed the whole $V$. Thus, it remains to show that $G$ also enjoys properties~\ref{It:cdot2}--\ref{It:Asucc}.

Let us start with property~\ref{It:Aprev}. The vertex $y \in V$ belongs to $V_n$ for some $n$. Using the bijectivity of $\sigma$, we conclude that there exists such $i$ that $3i \geq n$ and $\sigma(i) = ((y,A), k)$ (for some $k$). Therefore, $y \in V_{3i}$, and at the transition of type~0 from $G_{3i}$ to $G_{3i+1}$ we added a vertex $x$, the properties of which are exactly the ones required.
Property~\ref{It:Asucc} is symmetric, using a transition of type~1.

Finally, let us prove property~\ref{It:cdot2}. Let us first consider the case where $x \ne z$. Again, vertices $x$ and $z$ belong to some $V_n$. There exists such $i$ that $\varsigma(i) = ((x,z,B,C),k)$ and $3i+2 \geq n$. Since we indeed have $\vdash \ell(x,z) \to B \cdot C$ and $x \ne z$, the corresponding transition of type~2 is not skipped. This transition introduces $y$ with the desired properties. Indeed, $\ell(x,y) = \One \cdot A$  and $\ell(y,z) = B \cdot \One$, which  yields $\vdash \ell(x,y) \to B$ and $\vdash \ell(y,z) \to C$.

Now let $x = z$. Then $\ell(x,z) = \ell(x,x) = \One$, and we have $\vdash \One \to B \cdot C$. By cut with ${\to \One}$ (axiom), we get $\vdash {} \to B \cdot C$. Let us eliminate the cut rule in this proof.\footnote{As noticed in Section~\ref{S:calculus}, cut elimination for $\LEwuz$ is standard. Below, in the proof of Theorem~\ref{Th:counterexample}, we sketch the cut elimination proof for an extension of $\LEwuz$.} The lowermost rule in the cut-free proof is nothing but $\cdot R$. Therefore, we get $\vdash {} \to B$ and $\vdash {} \to C$.  In its turn, by $\One L$, this yields $\vdash \One \to B$ and $\vdash \One \to C$, or, in other words, $\vdash \ell(x,x) \to B$ and $\vdash \ell(x,x) \to C$. Thus, taking $y = x$ satisfies property~\ref{It:cdot2}.
Notice that this is the only place in the proof where we cannot allow extra axioms from $\Hc$ and fail to prove strong completeness.\footnote{This failure is actually even not due to the absence of cut elimination in the presence of $\Hc$. Indeed, one could just add $\One \to b \cdot c$ ($b$ and $c$ are variables) as an axiom, while $\One \to b$ and $\One \to c$ are not derivable. The extra axiom $\One \to b \cdot c$ can be reformulated as a good sequent calculus rule, see Section~\ref{S:counterexample} for more details.}
\end{proof}

\begin{rem}\label{Rem:algebras}
Andr\'eka and Mikul\'as~\cite{AndrekaMikulas1994} in their proof of Theorem~\ref{Th:AndrekaMikulas}  use a more abstract algebraic framework: labels on graph edges are not formulae but elements of a residuated semi-lattice (that is, algebraic model for $\LL$, see~\cite{GalatosRLbook} for details). In other words, strong completeness appears as a corollary of a purely algebraic representation theorem. In our case, however, the representation theorem holds only for the Lindenbaum--Tarski algebra, which consists of equivalence classes of formulae. Indeed, the representation theorem for arbitrary algebras would have yielded strong completeness, which does not hold (see Section~\ref{S:counterexample} below). Thus, it does not matter whether to use formulae (as we do) or elements of this algebra as labels.
\end{rem}

\begin{rem}\label{Rem:filter}
If one takes the reduct of a non-standard square R-model by removing constants~$\Zero$ and~$\One$, the result will be a square R-model in the usual sense.
In particular, this holds for our universal model $\Md_0^{\Af}$. 
Thus, we obtain an alternative proof of Mikul\'as' Theorem~\ref{Th:Mikulas}. In our construction, labels on edges are formulae, while Mikul\'as used filters, which are sets of formulae. This is achieved by using the explicit unit constant: without $\One$, there are incompatible formulae which should be labels of the same loop (e.g., $p \BS p$ and $q \BS q$ for different variables $p$ and $q$).  The tradeoff for this simplification is the use of a more complicated, non-standard interpretations of constants. Below (Section~\ref{S:iterative}) we extend this argument to an infinitary extension of $\LEw$.
\end{rem}

\section{Counter-Example to Strong Completeness}\label{S:counterexample}

Unlike the case with Lambek's restriction, for $\LEw$ strong completeness w.r.t.\ square R-models does not hold.
We prove one result for both $\LEw$ and its extension $\LEwuz$. For the latter, strong completeness fails w.r.t.\ the class of non-standard models defined in Section~\ref{S:constants}. A series of potential counterexamples to strong completeness was given by Mikul\'as~\cite{Mikulas2015Studia}. Here we prove that the first one of them is indeed such a counterexample.

\begin{thm}\label{Th:counterexample}
Let $a,b,c,d$ be distinct variables.
Then 
 $a \BS a \yields b \cdot c \vDash_{\text{\rm square R-models}} 
 d \yields d \cdot b \cdot \bigl((c \cdot b) \wedge (a \BS a) \bigr) \cdot c$, but not 
 $a \BS a \yields b \cdot c \vdash_{\LEwuz}  d \yields d \cdot b \cdot \bigl((c \cdot b) \wedge (a \BS a) \bigr) \cdot c$.
Therefore, neither $\LEw$ is strongly complete w.r.t.\ square R-models (since non-derivability in $\LEwuz$ implies that in $\LEw$), nor $\LEwuz$ is strongly complete w.r.t.\ non-standard square R-models (those sequents do not include constants, whence non-standardness of the models has no effect on their validity).
\end{thm}

\begin{proof}
 The first part (semantic entailment) is due to Mikul\'{a}s~\cite[Remark~5.3]{Mikulas2015Studia}. We reproduce it here in order to keep our exposition self-contained. Let us show that $(y,y) \in v(b \cdot ((c \cdot b) \wedge (a \BS a)) \cdot c)$ for any $y \in W$. Then for any $(x,y) \in v(d)$ we shall have $(x,y) \in v(d \cdot b \cdot ((c \cdot b) \wedge (a \BS a)) \cdot c)$.
 
 We have $(y,y) \in v(a \BS a)$, since $\delta \subseteq v(a) \BS v(a)$ for any $v(a)$. Therefore, since $a \BS a \yields b \cdot c$ is true in $\Md$, we get $(y,y) \in v(b) \circ v(c)$. This means that there exists such $z \in W$ that $(y,z) \in v(b)$ and $(z,y) \in v(c)$.
 In its turn, this gives $(z,z) \in v(c \cdot b)$; we also have $(z,z) \in v(a \BS a)$, therefore $(z,z) \in v((c \cdot b) \wedge (a \BS a))$.
 
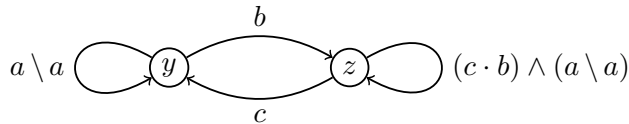
  \begin{center}
 \begin{tikzpicture}[line width=.7pt]
 \begin{scope}[scale=0.6]

\node[circle,inner sep=2pt,draw,minimum size=5mm] (a) at (0,0) {$y$};
\node[circle,inner sep=2pt,draw,minimum size=5mm] (b) at (4,0) {$z$};

\draw (a) edge[out=150,in=210,looseness=15,->] node[left] {$a \BS a$} (a);
\draw (a) edge[out=30,in=150,->] node[above] {$b$} (b);
\draw (b) edge[out=210,in=-30,->] node[below] {$c$} (a);
\draw (b) edge[out=30,in=-30,looseness=15,->] node[right] {$(c \cdot b) \wedge (a \BS a)$} (b);

\end{scope}
\end{tikzpicture}
\end{center}
  This yields $(y,y) \in v(b) \circ v((c \cdot b) \wedge (a \BS a)) \circ v(c)$, q.e.d.

 Now let us show that $d \yields d \cdot b \cdot \bigl((c \cdot b) \wedge (a \BS a) \bigr) \cdot c$ is not derivable from $a \BS a \yields b \cdot c$ in $\LEwuz$.  We do it in a syntactic way.  
 Suppose the contrary and let the sequent be derivable from our hypotheses.  In this derivation, let us substitute $\One$ for $a$ and $d$. This yields derivability of $\One \to \One \cdot b \cdot \bigl( (c \cdot b) \wedge (\One \BS \One) \bigr) \cdot c$ from $\One \BS \One \to b \cdot c$.
 
 Next, we notice that $\One \BS \One \to b \cdot c$ is derivable from ${\to b \cdot c}$ using $\One L$, $\One R$, and $\BS L$. On the other hand, given $\One \to \One \cdot b \cdot \bigl( (c \cdot b) \wedge (\One \BS \One) \bigr) \cdot c$, we apply cut with $\to \One$ (axiom) and $\One \cdot b \cdot \bigl( (c \cdot b) \wedge (\One \BS \One) \bigr) \cdot c \to b \cdot \bigl( (c \cdot b) \wedge \One \bigr) \cdot c$ (derivable in $\LEwuz$). This argument gives the following: $\yields b \cdot \bigl( (c \cdot b) \wedge \One \bigr) \cdot c$ is derivable in $\LEwuz$, extended with $\yields b \cdot c$ as an extra axiom.
 
 Let us introduce an auxiliary calculus $\LEwuzbc$, which is $\LEwuz$ extended with the following inference rule:
 \[
  \infer[bc]
  {\Gamma, \Delta \yields F}
  {\Gamma, b, c, \Delta \yields F}  
 \]
 (Notice that here $b$ and $c$ are {\em concrete} variables, not meta-symbols.)
 
 Adding this new rule to $\LEwuz$ is equivalent to adding $\yields b \cdot c$ as an axiom. Indeed, $\yields b \cdot c$ can be derived using the $bc$ rule and the $bc$ rule can be simulated using $\yields b \cdot c$ and cut.
 
The new calculus $\LEwuzbc$, however, enjoys cut elimination. The proof is standard (going back to Lambek's original paper~\cite{Lambek1958}) and proceeds by nested induction: (1) on the complexity of  the formula $A$ being cut; (2) on the height of the derivation tree above the cut.

At each step we consider the lowermost rules in the derivations of the premises of cut. The only new situation here is when at least one of these rules is $bc$; other cases are standard.

For $bc$ on the left, we propagate cut as follows:
\[
 \infer[Cut]{\Gamma, \Pi', \Pi'', \Delta \yields B}
 {\infer[bc]{\Pi',\Pi'' \yields A}
 {\Pi',b,c,\Pi'' \yields A} & 
 \Gamma, A, \Delta \yields B} 
\quad
\raisebox{1em}{$\leadsto$}
\quad
  \infer[bc]{\Gamma, \Pi', \Pi'', \Delta \yields B}
 {\infer[Cut]{\Gamma, \Pi', b,c,\Pi'', \Delta \yields B}
 {\Pi',b,c,\Pi'' \yields A & 
 \Gamma, A, \Delta \yields B}}
\]

For $bc$ on the right:
\[
\infer[Cut]{\Gamma',\Gamma'', \Pi,\Delta \yields B}
{\Pi \yields A & \infer[bc]{\Gamma', \Gamma'', A, \Delta \yields B}
{\Gamma', b,c, \Gamma'', A, \Delta \yields B}}
\quad\raisebox{1em}{$\leadsto$}\quad
 \infer[bc]{\Gamma',\Gamma'', \Pi,\Delta \yields B}
 {\infer[Cut]{\Gamma', b,c, \Gamma'', \Pi, \Delta \yields B}
 {\Pi \yields A & \Gamma', b,c, \Gamma'', A, \Delta \yields B}}
\]
and similarly in the case when $b,c$ appear in $\Delta$.

Now we may suppose that $\yields b \cdot \bigl( (c \cdot b) \wedge \One \bigr) \cdot c$ has a cut-free proof in $\LEwuzbc$. Let us track this proof from the goal sequent to the application of $\wedge R$ which introduces $(c \cdot b) \wedge \One$. Below this $\wedge R$ there are only two applications of $\cdot R$ and several ones of $bc$. Therefore, the application of $\wedge R$ 
derives $\Pi \to (c \cdot b) \wedge \One$ from $\Pi \to c \cdot b$ and $\Pi \to \One$,
where $\Pi$ is a sequence of $b$'s and $c$'s.

However, $\Pi \yields \One$, if $\Pi$ contains no connectives, is derivable only if $\Pi$ is empty. Thus, we get derivability of $\yields c \cdot b$. Let us again track its derivation up to the application of $\cdot R$, which derives $\Phi,\Psi \to c \cdot b$ from $\Phi \to c$ and $\Psi \to b$.
Here the sequence $\Phi,\Psi$ was obtained by several applications of the $bc$ rule. Therefore, either $\Phi$ is empty, or the first element of $\Phi$ is $b$. On the other hand, $\Phi \yields c$, where $\Phi$ is a sequence of $b$'s and $c$'s, is derivable only if $\Phi = c$. Contradiction.
\end{proof}

\section{An Infinite Conjunction: Iterative Division}\label{S:iterative}

Now let us return to our weak completeness result (Section~\ref{S:weak}). 
The construction used in its proof, where each edge is labelled by just one formula\footnote{In the setting of Mikula\'s~\cite{Mikulas2015Studia,Mikulas2015Synthese}, this formula generates a {\em principal} filter of all formulae which are ``valid'' on the edge.}, allows extension to {\em infinite} conjunctions.

In order to make our syntax simpler, we shall not consider arbitrary infinite conjunctions, since this would require development of an infinitary formula language. We concentrate on one important particular case, in which the formula language is finitary (while proofs could be infinitary). This particular case of infinitary conjunction is connected to the Kleene star. 

In relational models, the Kleene star is the operation of taking the reflexive-transitive closure of a relation. Thus, it can be represented as an infinite union: \[R^* = \delta \cup R \cup (R \circ R) \cup (R \circ R \circ R) \cup \ldots\] Adding a union-like connective, however, causes incompleteness issues connected with distributivity. A concrete corollary of the distributivity law, using meet and the Kleene star, which is not derivable without distributivity, is given in~\cite[Theorem~4.1]{Kuzn2018AiML}.

When put under division, however, the infinite union turns into an infinite intersection: \[S \SL R^* = 
S \cap (S \SL R) \cap ((S \SL R) \SL R) \cap \ldots,\] and similarly for $R^* \BS S$. Thus, instead of one unrestricted Kleene star, we consider two composite connectives: $A^* \BS B$ and $B \SL A^*$. Following Sedl\'ar~\cite{Sedlar2020}, who introduced similar connectives in a non-associative setting and with positive iteration (Kleene plus) instead of Kleene star (due to Lambek's restriction), we call these connectives {\em iterative divisions.} 
Independently from Sedl\'ar, such connectives were introduced in~\cite{KuznetsovRyzhkova2020}. The system considered there is associative, but still has Lambek's restriction, so Kleene plus is used instead of Kleene star. In~\cite{KuznetsovRyzhkova2020} it was proved that the Lambek calculus $\LL$ extended with meet and iterative divisions is strongly complete w.r.t.\ the class of all R-models. In this paper, we shall prove a weak counterpart of that result for the system without Lambek's restriction.

An infinitary proof system for the Lambek calculus with Kleene star, or infinitary action logic, was introduced by Buszkowski and Palka~\cite{Palka2007,BuszkowskiPalka2008}. We present a version of this system for iterative divisions, following~\cite{KuznetsovRyzhkova2020}:
\[
 \infer[*{\BS} L,\ n \ge 0]
 {\Gamma, \Pi_1, \ldots, \Pi_n, A^* \BS B, \Delta \to C}
 {\Pi_1 \to A & \ldots & \Pi_n \to A & \Gamma, B, \Delta \to C}
 \qquad
 \infer[*{\BS} R]
 {\Pi \to A^* \BS B}
 {\bigl( A^n, \Pi \to B \bigr)_{n=0}^{\infty}}
\]
\[
 \infer[{\SL}{*} L,\ n \ge 0]
 {\Gamma, B \SL A^*, \Pi_1, \ldots, \Pi_n, \Delta \to C}
 {\Pi_1 \to A & \ldots & \Pi_n \to A & \Gamma, B, \Delta \to C}
 \qquad
 \infer[{\SL}{*} R]
 {\Pi \to B \SL A^*}
 {\bigl( \Pi, A^n \to B \bigr)_{n=0}^{\infty}}
\]
The system obtained by adding these rules to $\LEwuz$ will be denoted by~$\LEwuzK$. 
A version of this system with Lambek's restriction is undecidable (namely, $\Pi^0_1$-complete)~\cite{KuznetsovRyzhkova2020}. For $\LEwuzK$, we also conjecture $\Pi^0_1$-completeness, thus using infinitary proof machinery (omega-rules or similar) becomes inevitable.

In square R-models, the Kleene star is interpreted as the reflexive-transitive closure operation:
\(
 v(A^*) = (v(A))^* \).
 Thus, the interpretation of iterative divisions is as follows: 
\(
 v(A^* \BS B) = (v(A))^* \BS v(B) \) and \( v(B \SL A^*) = v(B) \SL (v(A))^*\).

We extend the notion of non-standard square R-model with the unit (Definition~\ref{Df:RmodNS}) with this interpretation for iterative divisions. A routine check provides strong soundness. Notice that the usage of $\delta$ in the interpretation of the Kleene star does not conflict with the non-standard unit $\UA$, since they are equivalent in the denominator:
$R \SL \delta = R = R \SL \UA$ (see the proof of Lemma~\ref{Lm:deltaunit}).

Below we prove weak completeness. The strong one fails by Theorem~\ref{Th:counterexample}.
The reduct to the language without the unit yields ``standard'' square R-models with iterative divisions, thus we get soundness and weak completeness for them also.

\begin{thm}\label{Th:itdiv}
 If a sequent in the language with iterative divisions is true in all non-standard square R-models, then it is derivable in $\LEwuzK$.
\end{thm}

\begin{proof}
 This extension of Theorem~\ref{Th:weak} is proved in the same way as we do in~\cite{KuznetsovRyzhkova2020} for the case with Lambek's restriction. First, in Lemma~\ref{Lm:graph} we replace the set $\Fm$ of formulae used as labels by the one with iterative divisions. Thus, we get a new labelled graph $G$ using the same step-by-step construction (that is, we do not need to re-prove Lemma~\ref{Lm:graph}). Next, the only thing we need to modify in the proof of Theorem~\ref{Th:weak} is to add one more case, iterative division, in the check that $\Md_0^{\Af}$ is a well-defined model. Everything else remains the same.
 
 Thus, we have to prove that $v(A^* \BS B) = (v(A))^* \BS v(B)$ and $v(B \SL A^*) = v(B) \SL (v(A))^*$. We shall prove only the former, since the latter is symmetric. Let us first establish the $\subseteq$ inclusion.
 
  Suppose that $(y,z) \in v(A^* \BS B)$ and take an arbitrary $x \in W$ such that $(x,y) \in (v(A))^*$. Our aim is to show that $(x,z) \in v(B)$. The statement $(x,y) \in (v(A))^*$ means that there exists a number $n \geq 0$ and a sequence $x_0, x_1, \ldots, x_n \in W$ such that $x_0 = x$, $x_n = y$, and $(x_{i-1}, x_i) \in v(A)$ for each $i = 1,\ldots,n$. In particular, if $n = 0$, then we have $x = y$. For $n > 0$, we iterate property~\ref{It:cdot1} of $G$ and get $\vdash \ell(x,y) \to A^n$ and proceed as follows:
 \[
  \infer[Cut]
  {\ell(x,z) \to B}
  {\ell(x,z) \to \ell(x,y) \cdot \ell(y,z) & \infer[\cdot L]{\ell(x,y) \cdot \ell(y,z) \to B}
  {\infer[Cut]{\ell(x,y), \ell(y,z) \to B}{\ell(x,y) \to A^n & \infer[Cut]{A^n, \ell(y,z) \to B}{\ell(y,z) \to A^* \BS B & \infer[{*}{\BS}L]{A^n, A^* \BS B \to B}{\overbrace{A \to A\ \ldots}^{\text{$n$ times}}  & B \to B}}}}}
 \]
In the case of $n=0$ we have $\ell(x,z) = \ell(y,z)$, and the sequent $\ell(y,z) \to B$ is derived using cut with $\ell(y,z) \to A^* \BS B$ and the ${*}{\BS}L$ rule with $n = 0$.

Now let us establish the $\supseteq$ inclusion. Suppose that $(y,z) \in (v(A))^* \BS v(B)$. We need to show that $(y,z) \in v(A^* \BS B)$, that is, $\vdash \ell(y,z) \to A^* \BS B$. The latter is derived using the omega-rule ${*}{\BS}R$ from the infinite series of sequents $\bigl( A^n, \ell(y,z) \to B \bigr)_{n=0}^{\infty}$. For $n = 0$,  take $\delta \in (v(A))^*$ and conclude that $(y,z) \in v(B)$, thus, $\vdash \ell(y,z) \to B$. For $n > 0$, we iterate property~\ref{It:Aprev} of $G$ and construct a sequence  $x_0, x_1, \ldots, x_n$ such that $x_0 = y$, $(x_{i+1}, x_i) \in E$ and $\ell(x_{i+1},t) = A \cdot \ell(x_i,t)$ for any $t$ such that $(x_i,t) \in E$.
Having $\ell(x_0,z) = \ell(y,z)$ and $\ell(x_0,y) = \ell(y,y) = \One$, by induction we get $\ell(x_n,z) = A^n \cdot \ell(y,z)$ and $\ell(x_n,y) = A^n \cdot \One$. 

The latter yields $(x_n,y) \in v(A^n) = v(A) \circ \ldots \circ v(A) \subseteq (v(A))^*$. Thus, since $(y,z) \in (v(A))^* \BS v(B)$, we have $(x_n,z) \in v(B)$, that is, $\vdash \ell(x_n,z) \to B$.
Now the derivation of $A^n, \ell(y,z) \to B$ is as follows:
\[
 \infer[Cut]
 {A^n, \ell(y,z) \to B}
 {\infer[\cdot R]{A^n,\ell(y,z) \to \ell(x_n,z)}{A^n \to A^n & \ell(y,z) \to \ell(y,z)} & \ell(x_n,z)  \to B} \qedhere
\]
\end{proof}

\section{The Exponential Modality and Strong Conservativity}\label{S:exponential}

In the next section, we are going to prove a {\em strong} completeness result for  the product-free fragment of $\LEwuz$. This fragment will be denoted by $\LEDwuz$ and, as defined in Section~\ref{S:calculus}, it is obtained from $\LEwuz$ by simply removing rules for product, $\cdot L$ and $\cdot R$. Another interesting fragment is $\LEDw$, which also lacks axioms and rules for constants: $\One L$, $\One R$, $\Zero L$.

However, inside the proofs in Section~\ref{S:strong} {\em we  are going to use formulae with multiplication.} This makes {\em conservativity} an acute question. Otherwise, we would prove completeness not for $\LEDwuz$, but for a potentially bigger system in which, despite the set of hypotheses $\Hc$ and the goal sequent $\Pi \to B$ are product-free, the derivations are allowed to use product. So, let us formulate and prove the strong conservativity statement.

\begin{thm}\label{Th:strongconserv}
 If $\Hc$ and $\Pi \to B$ are in the product-free language (i.e., constructed using $\BS$, $\SL$, $\wedge$, $\Zero$, and $\One$) and $\Pi \to B$ syntactically follows from $\Hc$ in $\LEwuz$, then it also does so in $\LEDwuz$. Moreover, if $\Hc$ and $\Pi \to B$ do not include constants~$\Zero$ and~$\One$, then $\Pi \to B$ syntactically follows from $\Hc$ in $\LEDw$.
\end{thm}

Wishing to use cut elimination and the subformula property, we are going to {\em internalise} sequents from $\Hc$ into the goal sequent using a modalised version of the {\em deduction theorem.} (The deduction theorem itself does not hold for $\LEwuz$ due to the substructural nature of the calculus.) The modality which will be used for this purpose is the {\em exponential} borrowed from linear logic, see~\cite{Girard1987,LMSS,deGroote2005}. The exponential extension of $\LEwuz$ will be denoted by $\eLEwuz$.

The new system $\eLEwuz$ is obtained from $\LEwuz$ in the following way. The language of formulae is extended by the exponential as a unary connective, written in the prefix form: ${!}A$. The set of inference rules are extended by the following rules for the exponential:
\[
 \infer[{!}L]
 {\Gamma, {!}A, \Delta \yields C}
 {\Gamma, A, \Delta \yields C}
 \qquad
 \infer[{!}R]
 {{!}A_1, \ldots, {!}A_n \yields {!}B}
 {{!}A_1, \ldots, {!}A_n \yields B}
 \qquad
 \infer[{!}W]
 {\Gamma, {!}A, \Delta \yields C}
 {\Gamma, \Delta \yields C}
\]
\[
 \infer[{!}P_1]
 {\Gamma, {!}A, \Phi, \Delta \yields C}
 {\Gamma, \Phi, {!}A, \Delta \yields C}
 \qquad
 \infer[{!}P_2]
 {\Gamma, \Phi, {!}A, \Delta \yields C}
 {\Gamma, {!}A, \Phi, \Delta \yields C}
 \qquad
 \infer[{!}C]
 {\Gamma, {!}A, \Delta \yields C}
 {\Gamma, {!}A, {!}A, \Delta \yields C}
\]
Notice that ${!}$ allows all structural rules, namely weakening, permutation, and contraction. In other words, formulae of the form ${!}A$ behave like intuitionistic formulae rather than Lambek ones. This will allow reducing derivability from hypotheses in $\LEwuz$ to pure derivability in $\eLEwuz$.

The key feature of pure derivability in $\eLEwuz$ is {\em cut elimination:} any sequent derivable in $\eLEwuz$ is derivable without using $Cut$. The proof of cut elimination is a non-commutative variant of cut elimination in linear logic. For an accurate presentation of this proof, see~\cite{KKNS2019MSCS}. Our system $\eLEwuz$ is a fragment of the $\mathbf{SMALC}_\Sigma$ calculus considered in~\cite{KKNS2019MSCS}.

Now let us return to proving Theorem~\ref{Th:strongconserv}. For simplicity, we transform sequents in $\Hc$ into sequents with empty antecedents. Indeed, replacing $A_1, \ldots, A_n \yields B$ with $\yields A_n \BS \ldots (A_{n-1} \BS \ldots \BS (A_1 \BS B) \ldots)$ (as we already did in the proof of Theorem~\ref{Th:weak}) in $\Hc$ does not change derivability from $\Hc$. Moreover, this holds for each system considered, since the rules used here are only $Cut$, $\BS R$, and $\BS L$. Thus, from this point $\Hc$ consists of sequents with empty antecedents.

\begin{rem}
 It is important to notice that this construction, which makes antecedents empty, does not use multiplication. Thus, it can be performed in the narrower, product-free calculi $\LEDw$ and $\LEDwuz$.
\end{rem}

For a {\em finite} set of hypotheses $\Hc' = \{  {\to A_1}, \ldots, {\to A_n} \}$, let us define the following sequence of formulae in the language of $\eLEwuz$:
\[
 {!}\Hc' = \{ {!}A_1, \ldots, {!}A_n \}.
\]
The order of $A_i$'s here is arbitrary, since ${!}\Hc'$ will be used in antecedents where permutation rules are available for ${!}$-formulae.

Let us formulate and prove the version of deduction theorem we need for strong conservativity. The idea here is not new, and it goes back to~\cite{LMSS,KKNS2019MSCS}.

\begin{lem}\label{Lm:ded}
 Let $\Hc$ and $\Pi \yields B$ be, resp., a set of sequents (possibly infinite) and a sequent in the product-free language. Then the following holds.
 \begin{enumerate}
  \item\label{It:dedforward} If\/ $\Pi \yields B$ syntactically follows from $\Hc$ in $\LEwuz$, then there exists a finite $\Hc' \subseteq \Hc$ such that the following sequent is derivable in $\eLEwuz$:
  \[
   {!}\Hc', \Pi \yields B.
  \]
  \item\label{It:dedbackward} If the sequent ${!}\Hc', \Pi \yields B$ is derivable in $\eLEwuz$, then $\Pi \yields B$ syntactically follows from $\Hc'$ (and, therefore, from $\Hc$) in $\LEDwuz$.
  \item\label{It:dedbackward2} Moreover, if $\Hc$ and $\Pi \yields B$ do not contain constants, then 
  $\Pi \yields B$ syntactically follows from $\Hc'$ already in $\LEDw$.
 \end{enumerate}

\end{lem}

\begin{proof}
 For~(\ref{It:dedforward}), we proceed by induction on the derivation of $\Pi \to B$ from $\Hc$. For axioms $Id$, $\One R$, and $\Zero R$, we take $\Hc' = \varnothing$ and get just the same axiom. For a sequent of the form ${\yields A}$ from $\Hc$, which can be also used as an axiom, we take $\Hc' = \{ {\yields A} \}$ and obtain ${!}A \yields A$, which is derivable using ${!}L$. For one-premise rule, we just keep the same ${!}\Hc'$, possibly moving it using ${!}P_2$ (this is necessary for the $\BS R$ rule). Finally, for two-premise rules we have two subsets of $\Hc$, which we denote by $\Hc'_1$ and $\Hc'_2$.  The new $\Hc'$ is their union. We apply the same rule as in the original derivation, moving the ${!}$-formulae using permutation rules and in the end contracting duplicated formulae of the form ${!}A$ where $A \in \Hc'_1 \cap \Hc'_2$ (if any). In particular, this also happens for $Cut$ (which is a two-premise rule), so the resulting derivation in $\eLEwuz$ could include cuts. However, as noticed above, $Cut$ can be eliminated. Since the original derivation was finite, $\Hc'$ is also always finite. Namely, it includes all sequents from $\Hc$ actually used at least once.
 
 Statements~(\ref{It:dedbackward}) and~(\ref{It:dedbackward2}) are proved by the same argument. Consider a cut-free derivation of ${!}\Hc', \Pi \yields B$ in $\eLEwuz$. Cut-free derivations enjoy the subformula property, whence this derivation does not include rules for product, $\cdot L$ and $\cdot R$. For~(\ref{It:dedbackward2}), it also does not include axioms and rules for constants: $\One L$, $\One R$, $\Zero L$. Now let us erase all formulae of the form ${!}A$ from this derivation. This trivialises all rules for ${!}$, except ${!}L$ (the ${!}R$ rule is never used, since $!$-formulae never appear in succedents). For ${!}L$, we have $A \in \Hc'$, whence it translates to the following application of $Cut$:
 \[
  \infer[Cut]
  {\Gamma, \Delta \yields C}
  {{} \yields A & \Gamma, A, \Delta \yields C}
 \]
 (the instance of ${!}A$ in the conclusion was erased).  All other rules are translated directly to the corresponding rules of $\LEDwuz$ (for~(\ref{It:dedbackward}) or
 $\LEDw$ (for~(\ref{It:dedbackward2}). This gives the desired derivation of $\Pi \to B$ from $\Hc$.
\end{proof}

Notice that this argument does not work for infinitary systems, like our system with iterated divisions $\LEwuzK$. For this system, we leave strong conservativity an open question, as well as strong completeness of the product-free fragment (see next section).

\begin{rem}
 In fact, the weakening rule ${!}W$ was never used in the proof of Lemma~\ref{Lm:ded}. Therefore, the full-power exponential modality here may be replaced by the so-called {\em relevant} modality~\cite{KKS2016FG}, which has the same set of rules as the exponential, except ${!}W$. Indeed, our $\Hc'$ set collects exactly those sequents from $\Hc$ which are really used for deriving $\Pi \yields B$, thus, we are talking about entailment in the sense of relevant logic rather than, say, intuitionistic one.
\end{rem}

Now we are ready to prove strong conservativity.

\begin{proof}[Proof of Theorem~\ref{Th:strongconserv}]
 Let $\Pi \yields B$ semantically follow from $\Hc$ in $\LEwuz$. Then by statement (\ref{It:dedforward}) of  Lemma~\ref{Lm:ded}, the sequent ${!}\Hc', \Pi \yields B$ is derivable in $\eLEwuz$ for some finite $\Hc' \subseteq \Hc$. By statement (\ref{It:dedbackward}), $\Pi \yields B$ is derivable from $\Hc$ in the smaller system $\LEDwuz$. In the case where $\Hc$ and $\Pi \yields B$ also do not contain constants $\Zero$ and $\One$, we use statement (\ref{It:dedbackward2}) and get derivability of $\Pi \yields B$ from $\Hc$ in $\LEDw$.
\end{proof}

\section{Strong Completeness without Product}\label{S:strong}

In this section we are going to prove strong completeness results for the product-free fragments of $\LEwuz$ and $\LEw$. For the first one, denoted by $\LEDwuz$, strong completeness will be proved for a {\em modified} class of non-standard square R-models, which we call {\em product-free non-standard square R-models.} This class could look {\em ad hoc} and artificial, but using such models will allow proving strong completeness of $\LEDw$ (the product-free calculus without constants~$\Zero$ and~$\One$) w.r.t.\ square R-models in the standard sense.

In order to define product-free non-standard square R-models, we first relax the notions of $\Af$-unit and $\Af$-zero (Definition~\ref{Df:Aunit}) by allowing $\Af$ not to be closed under composition. (At the same time, $\Af$ should still be closed under $\BS$, $\SL$, $\cap$.) The definition of $\UA$ and $\ZA$ remains the same: $\UA \circ R = R \circ \UA = R$ and $\ZA \subseteq R$ for any $R \in \Af$. Uniqueness of $\UA$ and $\ZA$ (if they exist) still holds. 

As for properties of $\ZA$ and $\UA$, the situation is trickier. Lemma~\ref{Lm:deltaunit} still holds, with the same proof. (Note that Remark~\ref{Rem:outside} also applies here.) As for Lemma~\ref{Lm:zerozero}, we prove only one inclusion:

\begin{lem}\label{Lm:zerozero2}
 If $\ZA$ is the $\Af$-zero, where $\Af$ is not necessarily closed under composition, then $\ZA \circ R \subseteq \ZA$ and $R \circ \ZA \subseteq \ZA$ for any $R \in \Af$.
\end{lem}

\begin{proof}
 Exactly as the corresponding part of the proof of Lemma~\ref{Lm:zerozero}. 
\end{proof}

\begin{cor}\label{Cor:zerozero2}
 If $\ZA$ is the $\Af$-zero, where $\Af$ is not necessarily closed under composition,
 then $R_1 \circ \ldots \circ R_m \circ \ZA \circ S_1 \circ \ldots \circ S_k \subseteq \ZA$ for any $R_1, \ldots, R_m, S_1, \ldots, S_k \in \Af$.
\end{cor}

\begin{proof}
 By Lemma~\ref{Lm:zerozero2}, using monotonicity of $\circ$ w.r.t.\ $\subseteq$.
\end{proof}

The definition of our modified models, product-free non-standard square R-models, is the same as Definition~\ref{Df:RmodNS}, but $\Af$ is not required to be closed under composition, and the $\Af$-unit and $\Af$-zero are understood also in this relaxed version of definition.

\begin{prop}
 The calculus $\LEDwuz$ is strongly sound w.r.t.\ the class of product-free non-standard square R-models.
\end{prop}

\begin{proof}
 Again, we basically repeat the proof of Proposition~\ref{Prop:sound}. First we notice that for any product-free formula $A$ we have $v(A) \in \Af$. (For formulae with product, this is now, in general, not true.) This allows considering the cases of $\One R$ and $\One L$ exactly as in the proof of Proposition~\ref{Prop:sound}, using Lemma~\ref{Lm:deltaunit}. For $\Zero L$, we use Corollary~\ref{Cor:zerozero2} and show that the interpretation $R$ of the antecedent is a subset of $\ZA$. Since $C$, the succedent, is product-free, we have $R \subseteq \ZA \subseteq v(C)$.
\end{proof}

Notice that the original calculus $\LEwuz$ is, in general, not sound w.r.t.\ these models even in the weak sense. For example,  Lemma~\ref{Lm:zerozero2} lacks the other inclusion, $\ZA \subseteq \ZA \circ R$, which prevents from validating the sequent $\Zero \yields \Zero \cdot r$. Fortunately, this sequent is not a product-free one.

Now we are ready to prove strong completeness for $\LEDwuz$.

\begin{thm}\label{Th:strong}
 The calculus $\LEDwuz$ is strongly complete w.r.t.\ the class of product-free non-standard square R-models.
\end{thm}

As mentioned in the previous section, inside our proofs we shall use multiplication. Due to this reason, the set $\Fm$ still denotes the set of all $\LEwuz$ formulae. Let us fix a set $\Hc$ of product-free sequents and prove a weaker version of Lemma~\ref{Lm:graph} for derivability from $\Hc$. Namely, in the presence of a non-empty set of hypotheses we are unable to maintain property~\ref{It:cdot2} of the graph, so we just omit it. By $\Hc \vdash A \to B$ we mean ``$A \to B$ is derivable from $\Hc$ in $\LEwuz$.''

\begin{lem}\label{Lm:graphH}
There exists  a labelled directed graph $G = (V,E,\ell)$ where $V \ne \varnothing$, $E \subseteq V \times V$, and $\ell \colon E \to \Fm$, such that the following holds:
\begin{enumerate}
  \item $E$ is transitive;
  \item $E$ is reflexive and $\ell(x,x) = \One$ for any $x \in V$;
  \item $E$ is antisymmetric: if $x \ne y$ and $(x,y) \in E$, then $(y,x) \notin E$;
  \item if $(x,y) \in E$ and $(y,z) \in E$, then $\Hc \vdash \ell(x,z) \yields \ell(x,y) \cdot \ell(y,z)$;
  \item for any $y \in V$ and any formula $A$ there exists $x \in V$ such that 
  for any $z \in V$, if $(y,z) \in E$, then $\ell(x,z) =  A \cdot \ell(y,z)$; 
  \item for any $y \in V$ and any formula $A$ there exists  $z \in V$ such that for any $x \in V$, if $(x,y) \in E$, then $\ell(x,z) = \ell(x,y) \cdot A$.
 \end{enumerate}
\end{lem}

\begin{proof}
 We go through the same infinite procedure as in the proof of Lemma~\ref{Lm:graph}, but apply only transitions of types~0 and~1, not of type~2. As transitions of type~2 were necessary exactly for obtaining property~\ref{It:cdot2} of graph $G$, without them we shall obtain a graph obeying all other properties, w.r.t.\ derivability from $\Hc$ instead of pure derivability.
\end{proof}

Now we prove strong completeness for $\LEDwuz$.

\begin{proof}[Proof of Theorem~\ref{Th:strong}]

Given a set of hypotheses $\Hc$, let us construct an {\em $\Hc$-universal} product-free non-standard square R-model $\Md_{\Hc}^{\Af} = (W, \Af, \UA, \ZA, v_{\Hc})$ in a way similar to how we did it in the proof of Theorem~\ref{Th:weak}, using the graph~$G$ given by Lemma~\ref{Lm:graphH}. 

We take $W = V$ and define $v_{\Hc}(A) = \{ (x,y) \in E \mid \Hc \vdash \ell(x,y) \yields A \}$ where $A$ is a variable or constant ($\Zero$ or $\One$). Below we shall prove that this equality holds for any {\em product-free} formula $A$. Then we may define $\Af = \{ v_{\Hc}(A) \mid \mbox{$A$ is a product-free formula} \}$ and assert that $\Af$ is closed under $\BS$, $\SL$, and $\cap$ (but possibly not under $\circ$). Finally, we shall show that $v_{\Hc}(\One)$ and $v_{\Hc}(\Zero)$ are indeed the $\Af$-unit and the $\Af$-zero, respectively.

The fact that $v_{\Hc}(A) = \{ (x,y) \mid \Hc \vdash \ell(x,y) \yields A \}$ is proved by induction on the structure of $A$ exactly as in the proof of Theorem~\ref{Th:weak}. Indeed, now we do not need to handle multiplication, so the omitted property of graph $G$ is not used. As for the differences connected to $\Af$ (which is now not closed under composition), this proof has nothing to do with $\Af$, so these differences are irrelevant. Finally, the shift from pure derivability to derivability from $\Hc$ does not affect the argument.

Finally, we check that $v_{\Hc}(\One)$ is the $\Af$-unit and $v_{\Hc}(\Zero)$ is the $\Af$-zero. The latter is simpler: $v_{\Hc}(\Zero) = \{ (x,y) \mid \Hc \vdash \ell(x,y) \yields \Zero \} \subseteq \{ (x,y) \mid \Hc \vdash \ell(x,y) \yields A \} = v_{\Hc}(A)$ for any product-free formula~$A$. 

For the former, formally we cannot use the argument from the proof of Theorem~\ref{Th:weak}, since it uses $v_{\Hc}(\One \cdot A)$, and $\One \cdot A$ is not product-free. Let us provide an explicit argument for $v_{\Hc}(\One) \circ v_{\Hc}(A) = v_{\Hc}(A)$ (the other equality, $v_{\Hc}(A) \circ v_{\Hc}(\One) = v_{\Hc}(A)$, is symmetric). First, we have $\ell(x,x) = \One$ for any $x$, whence $\delta \subseteq v_{\Hc}(\One)$. This establishes one inclusion: $v_{\Hc}(A) = \delta \circ v_{\Hc}(A) \subseteq v_{\Hc}(\One) \circ v_{\Hc}(A)$. For the opposite inclusion, let $(x,y) \in v_{\Hc}(\One) \circ v_{\Hc}(A)$. This means that, for some $w$, we have $(x,w) \in v_{\Hc}(\One)$ and $(w,y) \in v_{\Hc}(A)$. Thus, we have $\Hc \vdash \ell(x,w) \yields \One$, $\Hc \vdash \ell(w,y) \yields A$, and, by property~\ref{It:cdot1} of graph $G$,  $\Hc \vdash \ell(x,y) \yields \ell(x,w) \cdot \ell(w,y)$. Now we give the following derivation:
\[
 \infer[Cut]
 {\ell(x,y) \yields A}
 {\ell(x,y) \yields \ell(x,w) \cdot \ell(w,y)  & \infer[Cut]{\ell(x,w) \cdot \ell(w,y) \yields A}{\infer[\cdot L]{\ell(x,w) \cdot \ell(w,y) \yields \One \cdot A}{\infer[\cdot R]{\ell(x,w), \ell(w,y) \yields \One \cdot A}{\ell(x,w) \yields \One & \ell(w,y) \yields A}} & \infer[\cdot L]{\One \cdot A \yields A}{\infer[\One L]{\One, A \yields A}{A \yields A}}}}
\]
which establishes $(x,y) \in v_{\Hc}(A)$.

\vskip 3pt
Now, using our $\Hc$-universal model $\Md_{\Hc}^{\Af}$, we show strong completeness. Again, we restrict ourselves to sequents with empty antecedents, both in $\Hc$ and as the goal sequent. Suppose $\Hc$ semantically entails a sequent of the form $\yields B$ on the class of product-free non-standard square R-models. 

Let us show that all sequents from $\Hc$ are true in $\Md_{\Hc}^{\Af}$. Take one such sequent, $\yields A$. For any $x$, we have $\ell(x,x) = \One$, and $\Hc \vdash \One \yields A$ by application of $\One L$. This means $(x,x) \in v_{\Hc}(A)$ for any $x$, that is, $\delta \subseteq v_{\Hc}(A)$. 

Now, by semantic entailment, $\yields B$ is also true in $\Md_{\Hc}^{\Af}$, that is, $\delta \subseteq v_{\Hc}(B)$. This means $\Hc \vdash \One \yields B$ (since $W$ is non-empty, we take an arbitrary $x$ with $(x,x) \in \delta$ and $\ell(x,x) = \One$), and by cut with ${} \yields \One$ we get $\Hc \vdash {} \yields B$.

The final issue is that derivability here is understood as derivability in $\LEwuz$. Using Theorem~\ref{Th:strongconserv} (strong conservativity), we transform it into derivability in $\LEDwuz$.
\end{proof}

Finally, as promised above, as a corollary we obtain strong completeness of $\LEDw$ w.r.t.\ square R-models, in the most standard sense (Definition~\ref{Df:Rmod}), cf. Remark~\ref{Rem:filter} above. 

\begin{thm}
 The calculus $\LEDw$ is strongly complete w.r.t.\ the class of square R-models.
\end{thm}

\begin{proof}
 Let a sequent $\to B$ semantically follow from $\Hc$ on square R-models, provided that $\to B$ and $\Hc$ do not include $\cdot$, $\One$, and $\Zero$. Then we take the $\Hc$-universal model $\Md_{\Hc}^{\Af}$ from the proof of Theorem~\ref{Th:strong} and consider its reduct to the language without constants. This reduct is a standard square R-model which satisfies all sequents from $\Hc$. Therefore, it also validates $\to B$, whence $\Hc \vdash {} \to B$. Using Theorem~\ref{Th:strongconserv}, we transform derivability in $\LEwuz$ into derivability in $\LEDw$.
\end{proof}

Along with Theorem~\ref{Th:AndrekaMikulasLE} by Andr\'eka and Mikul\'as this shows that $\cdot$ or $\wedge$ on its own does not leads to failure of strong completeness of $\LEw$, but together they do.

\section*{Concluding Remarks}

The completeness and incompleteness results presented in this article reveal the following peculiarities of relational semantics for the Lambek calculus and its extensions. First, Lambek's non-emptiness restriction indeed matters, and without it strong completeness surprisingly fails when the system includes both multiplicative and additive conjunctions (i.e., product and intersection). Weak completeness still holds. One of the conjunctions is fine for strong completeness. Second, adding one or both explicit constants~$\Zero$ and~$\One$ ruins completeness even in the weak sense. Fortunately, there are easy and arguably natural ways of modifying the models which restore completeness.

The counterexamples to weak completeness with constants are actually very deep ones. They show that these constants allow building constructions (e.g., $\Zero \SL (\Zero \SL A)$, $\One \wedge B$, or $\One \SL (C \SL C)$) which obey structural rules not available for arbitrary formulae. In particular, formulae of the form $\Zero \SL (\Zero \SL A)$ even obey classical Boolean logic! This suggests that axiomatising the {\em standard} interpretations for these constants is probably a hard problem. Here one would have to somehow hybridise the Lambek calculus itself and another calculus with more structural rules. (For comparison, see~\cite{KuznJANCL} where such attempts were made for the case of language semantics.)

From the point of view of complexity, we see that these constructions using constants behave much like exponential or subexponential modalities (i.e., enable structural rules, including contraction), which would probably make the complete theories for standard interpretations undecidable, cf.~\cite{LMSS,ChvalovskyHorcik2016}. These undecidability conjectures are supported by the analogous situation with constant~$\One$ in language models, where undecidability was established~\cite{KanKuzSce2021IC}. Even more interesting things in complexity could happen if we add some sort of infinitary operations, like Kleene star or iterated divisions, cf.~\cite{KuzSpeAPAL}. By now, we leave these questions aside as an interesting area of further research.

\bibliographystyle{alphaurl}
\bibliography{Rmod}

\end{document}